\documentclass{article}
\usepackage{iclr2027_conference,times}
\iclrfinalcopy

\usepackage[utf8]{inputenc}
\usepackage[T1]{fontenc}
\usepackage{amsmath,amssymb,amsthm}
\usepackage{booktabs,multirow,array}
\usepackage{graphicx}
\usepackage{xcolor}
\usepackage{microtype}
\usepackage{subcaption}
\usepackage{float}
\usepackage{algorithmic}
\usepackage{tikz}
\usetikzlibrary{arrows.meta,positioning,calc,fit,backgrounds,patterns}
\usepackage{pgfplots}
\pgfplotsset{compat=1.17}
\usepackage{hyperref}
\usepackage{url}
\usepackage{xurl}
\usepackage[nameinlink,noabbrev]{cleveref}

\newfloat{algorithm}{t}{loa}
\floatname{algorithm}{Algorithm}
\crefname{algorithm}{Algorithm}{Algorithms}

\newtheorem{proposition}{Proposition}

\crefname{proposition}{Proposition}{Propositions}



\newcommand{\drain}{\textsc{DrainSinkhorn}}
\newcommand{\R}{\mathbb{R}}
\newcommand{\ones}{\mathbf{1}}
\newcommand{\eps}{\varepsilon}
\newcommand{\LSE}{\operatorname{LSE}}
\newcommand{\expmone}{\operatorname{expm1}}

\newcommand{\Dmax}{D_{\max}}
\newcommand{\Rwork}{R_{\mathrm{work}}}
\newcommand{\Aset}{\mathcal{A}}

\newcounter{pendingcount}

\AtEndDocument{\typeout{PENDING PLACEHOLDERS: \arabic{pendingcount}}}

\newcommand{\resGPU}{A100-SXM4-80GB}
\newcommand{\resOuterBlocks}{10}

\newcommand{\resTorch}{2.5.1}
\newcommand{\resCUDA}{12.4}
\newcommand{\resPOTVersion}{0.9.6}
\newcommand{\resGeomLossVersion}{0.3.1}
\newcommand{\resOTTVersion}{0.6.0}

\newcommand{\resDepthMinMOT}{5}
\newcommand{\resDepthMedMOT}{85}
\newcommand{\resDepthMaxMOT}{500}                    %
\newcommand{\resWastedWorkMOT}{73}
\newcommand{\resRworkRange}{1.19--4.70}

\newcommand{\resMainSpeedupRange}{2.73--10.75}       %
\newcommand{\resMainOTTMax}{10.75}                %
\newcommand{\resMainPOTMax}{34}                   %
\newcommand{\resPOTOverSolve}{6--89}               %
\newcommand{\resMainPOTRange}{5.3--34}              %
\newcommand{\resAbsOTTTrain}{1.05--2.8}
\newcommand{\resAbsOfficial}{9.9--24}

\newcommand{\resPGLoss}{OTT-JAX/\drain{} down to 0.42}                 %

\newcommand{\resSepRange}{11.1--17.6}                %
\newcommand{\resSepDDfour}{64.5}                     %

\newcommand{\resRetireWallRange}{1.10--1.60}         %
\newcommand{\resRetireSolveRange}{1.15--2.01}        %
\newcommand{\resRetireOneWaveWall}{1.03}             %

\newcommand{\resMaskedDDsixteen}{1.33}               %
\newcommand{\resMaskedWin}{1.88}                     %
\newcommand{\resStaticDDsixteen}{1.15}               %
\newcommand{\resStaticWin}{2.27}                     %

\newcommand{\resRegroupSmallB}{0.99--1.02}           %
\newcommand{\resRegroupHom}{1.07}                    %
\newcommand{\resRegroupHet}{1.21}                    %

\newcommand{\resAuditWorst}{1.85\cdot10^{-3}}        %
\newcommand{\resProfLaunchDrain}{70}
\newcommand{\resProfLaunchOTT}{2504}
\newcommand{\resProfLaunchPOT}{11530}
\newcommand{\resLargeSlower}{7}
\newcommand{\resLargeFaster}{1.28 and 1.57}         %
\newcommand{\resLargeHighDim}{10.9--16.7}            %
\newcommand{\resReplayOTT}{1.62--1.74}
\newcommand{\resReplayPOT}{2.88--2.98}

\newcommand{\resTolOTTRange}{16--38}                   %
\newcommand{\resTolPOTRange}{32--94}                   %
\newcommand{\resTolOTTddfour}{29.5}                   %
\newcommand{\resTolOTTmotsixteen}{37.9}                %
\newcommand{\resQOne}{1.36}                            %
\newcommand{\resMemMhcDrain}{1.07}                     %
\newcommand{\resMemMhcTorch}{3.76}                     %
\newcommand{\resMemMhcCompiled}{2.84}                  %
\newcommand{\resMemPatchOfficial}{906}                 %
\newcommand{\resMemPatchDrain}{53}                     %
\newcommand{\resBucketPgMs}{3.47}                      %
\newcommand{\resBucketPg}{1.84}                        %
\newcommand{\resOTApgDrainMs}{1.88}  %
\newcommand{\resOTApgOfficial}{13.0}  %
\newcommand{\resOTApgPadded}{6.31}  %
\newcommand{\resOTApgOTT}{1.05}  %
\newcommand{\resOTApgBatching}{2.06}  %
\newcommand{\resOTApgShareOfficial}{16.3}  %
\newcommand{\resOTApgShareDrain}{4.7}  %
\newcommand{\resOTApgSinkShareOfficial}{12.8}  %
\newcommand{\resOTApgSinkShareDrain}{1.1}  %
\newcommand{\resOTApgStep}{1.13}  %
\newcommand{\resOTADevDrain}{6.4\cdot10^{-6}}  %
\newcommand{\resOTADevOthers}{0.14--0.21}  %
\newcommand{\resMHCCompile}{2.60--2.95}  %
\newcommand{\resMHCTriton}{2.54--7.91}  %
\newcommand{\resMHCTileF}{1.47--1.48}  %
\newcommand{\resMHCTileGrad}{8--46}  %
\newcommand{\resMHCTenUnconverged}{97}  %
\newcommand{\resMHCTenDrainColumn}{7\cdot10^{-6}}  %
\newcommand{\resPatchInfer}{9.94}  %
\newcommand{\resPatchTrain}{10.4}  %
\newcommand{\resPatchTrainShareOfficial}{16.4}  %
\newcommand{\resPatchTrainStep}{1.18}  %
\newcommand{\resPatchPairs}{32}  %
\newcommand{\resPatchRegistered}{30}  %
\newcommand{\resPatchValid}{16}  %
\newcommand{\resPatchPeers}{2.80--11.9}  %

\newcommand{\resEtoEMHCCompile}{1.32}  %
\newcommand{\resEtoEMHCDrain}{1.73}  %
\newcommand{\resEtoEMHCTile}{1.80}  %
\newcommand{\resEtoEMHCStepZero}{2\cdot10^{-5}}  %
\newcommand{\resEtoEMHCFinal}{7.694--7.697}  %
\newcommand{\resDMShareOfficial}{50}  %
\newcommand{\resDMShareDrain}{24}  %
\newcommand{\resDMOT}{3.35}  %
\newcommand{\resDMStep}{1.51}  %
\newcommand{\resDMUnderCells}{37}  %
\newcommand{\resDMUnderMassMax}{30}  %
\newcommand{\resDMDrainFP}{1.9}  %
\newcommand{\resDMOfficialFP}{616--4828}  %

\newcommand{\gpuOnChipRange}{15--80}
\newcommand{\gpuReleaseResRange}{1.00--1.55}

\newcommand{\gpuOTTA}{2.64--11.1}
\newcommand{\gpuOTTH}{2.21--11.6}
\newcommand{\gpuOTTL}{1.70--13.4}
\newcommand{\gpuPOTRange}{5.18--47.7}

\newcommand{\gpuMHCCompileF}{1.58--4.25}
\newcommand{\gpuMHCCompileFB}{2.26--7.55}
\newcommand{\gpuMHCTritonF}{14.3--52.3}
\newcommand{\gpuMHCTritonFB}{3.85--22.3}
\newcommand{\gpuMHCTileFwdErr}{1.00}
\newcommand{\gpuMHCTileGradErr}{5--2.6\cdot10^{4}}
\newcommand{\gpuMHCDrainErr}{3\cdot10^{-6}}
\newcommand{\longExactRange}{4.55--4.58}
\newcommand{\longTileRange}{6.07--6.37}
\newcommand{\longTileGrad}{116--83{,}279}
\newcommand{\longExactGrad}{0.57--0.69}

\definecolor{drainblue}{HTML}{2C6FBB}
\definecolor{drainorange}{HTML}{D9822B}
\definecolor{draingray}{HTML}{8C8C8C}
\definecolor{draingreen}{HTML}{3A9A5B}
\definecolor{drainlight}{HTML}{E8EEF6}

\tikzset{
  box/.style={draw=black!70, rounded corners=2pt, fill=white, inner sep=3pt, font=\scriptsize, align=center},
  sbox/.style={box, fill=drainlight},
  arr/.style={-{Stealth[length=4pt]}, black!70, line width=0.5pt},
  work/.style={fill=drainblue!75, draw=none},
  idle/.style={fill=draingray!35, draw=none},
  wasted/.style={pattern=north east lines, pattern color=drainorange, draw=drainorange!80},
}

\title{\drain: Instance-Resident Batched\\Entropic Optimal Transport on GPUs}

\author{Xinyang Wen}

\begin{document}

\maketitle
\lhead{Preprint}

\begin{abstract}
Many learning pipelines call entropic optimal transport (OT) inside a loop, and each call is a batch of small, independent problems: one per image in detection label assignment, per patch pair in point-cloud registration, per frame pair in tracking, and per tile in domain decomposition. General GPU solvers launch every Sinkhorn update over the whole batch and run the batch until its slowest problem stops, so practitioners write a kernel for each shape or drop OT. We present \drain, a batched Sinkhorn solver that makes each problem of the batch, an instance, the unit of GPU execution. Wherever an instance fits on chip, one GPU thread block solves it from start to finish and exits as soon as the instance converges, freeing its slot for a pending instance. With a fixed iteration count, several instances share a thread block, and the forward map and its exact gradient take one launch each. Larger instances leave later launches once they converge. On tracking and domain-decomposition workloads, \drain{} is \resMainSpeedupRange$\times$ faster than OTT-JAX, the fastest general batched GPU solver in our comparison, and \resMainPOTRange$\times$ faster than POT. The gain is largest on calls with thousands of instances. In training layers it is \resAbsOTTTrain$\times$ faster than OTT-JAX and \resAbsOfficial$\times$ faster than the Sinkhorn modules of OTA and GeoTransformer, cutting label assignment from \resOTApgShareOfficial\% to \resOTApgShareDrain\% of each GPU's training step and making GeoTransformer's patch matching \resPatchTrain$\times$ faster in a training step.

\end{abstract}

\section{Introduction}
\label{sec:intro}

Many pipelines call entropic optimal transport (OT) inside a loop, and each call is a batch of small, independent problems. Much of the recent use of Sinkhorn in machine learning takes this form. Manifold-constrained hyper-connections (mHC) solve one problem per token and layer \citep{xie2025mhc}, and DeepSeek-V4 trains with them at up to 1.6T parameters \citep{deepseek2026v4}. Distilling large language models solves one per generated token \citep{na2026wasd}, permutation-robust LLM training one per training prompt \citep{chen2025pearl}, supervised graph prediction one per predicted and target graph \citep{mendez2026graph}, learned visual attribution one per image \citep{schinagl2026rank}, and substructure alignment one per protein pair \citep{wang2026plasma}. At inference, a multi-object tracker couples the detections of consecutive frames, domain decomposition splits a large problem into cell problems that are ``numerous, small, and with diverse shapes'' \citep{medina2025flow,bonafini2021domain}, and coarse-to-fine registration solves one problem per matched patch pair \citep{yu2021cofinet,qin2022geotransformer}. In training, the Sinkhorn iterations are themselves a layer: mHC projects one $4\times4$ residual mixing matrix per token and layer with 20 iterations \citep{xie2025mhc}, Optimal Transport Assignment (OTA) assigns detection labels per image with 50 \citep{ge2021ota}, and GeoTransformer matches 256 patch pairs of $65\times65$ points per 3DMatch scan pair with 100 \citep{qin2022geotransformer}. Each step thus produces $B$ problems, each with its own cost, marginals, and regularization, that must all be solved before the step continues. Similar scalability pressures arise in other large-scale learning pipelines. In graph-based recommendation, linear-time GNNs bring training cost down to that of matrix factorization \citep{zhang2024ltgnn}.

These calls are expensive enough that practitioners work around them. They write kernels for one shape, such as the fused $n=4$ Sinkhorn kernel of mHC \citep{xie2025mhc}, the warp-level $4\times4$ solver of \citet{wang2026birkhoff}, and the CUDA cell solver of \citet{medina2025flow}, or they drop OT: later hyper-connection designs replace mHC's Sinkhorn step \citep{liu2026spectral}, supervised graph prediction learns to predict the matching because ``the inner Sinkhorn rollouts dominate training time'' \citep{mendez2026graph}, YOLOX replaces OTA's Sinkhorn by a top-$k$ rule because Sinkhorn adds 25\% training time \citep{ge2021yolox}, and LightGlue replaces SuperGlue's Sinkhorn as expensive in compute and memory \citep{lindenberger2023lightglue}. Where OT stays, it costs: Sinkhorn roughly doubles the training time of LLM distillation \citep{na2026wasd}. We measured the same effect: with a COCO-trained detector on an A100, OTA's label assignment takes \resOTApgShareOfficial\% of a training step on each GPU, \resOTApgSinkShareOfficial\% in its Sinkhorn loop (\cref{sec:inloop}). General GPU solvers organize the work around the batch \citep{flamary2021pot,cuturi2022ott,feydy2019sinkhorn}: each half-step, one of the two Sinkhorn updates, is a launch over all problems, the batch runs until its slowest problem stops, and unrolled differentiation, where used, records every iteration. For many small problems this organization wastes time in two ways. One update does little arithmetic compared with launching it and synchronizing with the host, and problems need very different numbers of iterations: in the MOT17 association workload of \cref{sec:experiments}, the iterations needed to reach a two-sided tolerance of $10^{-4}$ range from \resDepthMinMOT{} to the cap of \resDepthMaxMOT{} with a median of \resDepthMedMOT, so \resWastedWorkMOT\% of the updates in a batch run to its slowest problem go to problems that already pass (\cref{fig:teaser}).

\begin{figure}[t]
\centering
\begin{tikzpicture}[x=0.22cm, y=0.34cm]
  \node[font=\scriptsize\bfseries, anchor=west] at (-3.4,4.4) {(a) Fixed-width batch};
  \foreach \slot in {1,2,3} {
    \node[font=\tiny, anchor=east] at (-0.3,4-\slot) {slot \slot};
  }
  \foreach \slot/\d/\p in {1/2/1, 2/5/2, 3/3/3} {
    \fill[work] (0,4-\slot-0.35) rectangle (\d,4-\slot+0.35);
    \fill[wasted] (\d,4-\slot-0.35) rectangle (8,4-\slot+0.35);
    \node[font=\tiny, text=white] at (0.95,4-\slot) {$q_{\p}$};
  }
  \foreach \slot/\d/\p in {1/8/4, 2/2/5, 3/4/6} {
    \fill[work] (8,4-\slot-0.35) rectangle (8+\d,4-\slot+0.35);
    \ifnum\d<8 \fill[wasted] (8+\d,4-\slot-0.35) rectangle (16,4-\slot+0.35); \fi
    \node[font=\tiny, text=white] at (8.95,4-\slot) {$q_{\p}$};
  }
  \draw[black!60] (0,0.35) -- (16.5,0.35);
  \foreach \x in {0,8,16} \node[font=\tiny, anchor=north] at (\x,0.35) {\x};

  \begin{scope}[xshift=4.75cm]
  \node[font=\scriptsize\bfseries, anchor=west] at (-0.6,4.4) {(b) With release};
  \foreach \slot/\s/\d/\p in {1/0/2/1, 1/2/8/4, 2/0/5/2, 2/5/4/6, 3/0/3/3, 3/3/2/5} {
    \fill[work] (\s,4-\slot-0.35) rectangle (\s+\d,4-\slot+0.35);
    \draw[white, line width=0.6pt] (\s,4-\slot-0.35) -- (\s,4-\slot+0.35);
    \node[font=\tiny, text=white] at (\s+0.95,4-\slot) {$q_{\p}$};
  }
  \draw[black!60] (0,0.35) -- (16.5,0.35);
  \foreach \x in {0,10,16} \node[font=\tiny, anchor=north] at (\x,0.35) {\x};
  \draw[dashed, draingreen, thick] (10,0.35) -- (10,3.4);
  \node[font=\tiny, text=draingreen, anchor=west] at (10.2,2.0) {batch done};
  \end{scope}

  \fill[work] (0,-1.35) rectangle (1.2,-0.95);
  \node[font=\tiny, anchor=west] at (1.4,-1.15) {update work until the instance passes the check};
  \fill[wasted] (0,-2.15) rectangle (1.2,-1.75);
  \node[font=\tiny, anchor=west] at (1.4,-1.95) {work on instances that already pass};
\end{tikzpicture}
\hfill
\begin{minipage}[b]{4.3cm}\centering
{\scriptsize\bfseries (c) Depths in a real batch}\\[2pt]
\begin{tikzpicture}
\begin{axis}[width=4.6cm, height=3.3cm, ybar interval, xmin=0, xmax=500, ymin=0,
  xtick={0,100,200,300,400,500}, xticklabel style={font=\tiny}, yticklabel style={font=\tiny},
  xlabel={first-passage depth $D_b$}, ylabel={instances}, label style={font=\scriptsize},
  axis line style={black!60}, tick style={black!60}, x tick label as interval=false]
\addplot[fill=drainblue!75, draw=white] table[x=left, y=count, col sep=comma] {figures/data/plot_fig1c.csv};
\end{axis}
\end{tikzpicture}
\end{minipage}
\caption{\textbf{Instances finish at different depths, and fixed-width execution pays for the slowest one.}
(a)~Six instances with depths $(2,5,3,8,2,4)$ on three slots; a fixed-width batch runs every instance as long as the slowest one in the batch (hatched).
(b)~With release, a finished instance frees its slot at once for a pending one: 10 instead of 16 time units (illustration).
(c)~First-passage depths, the iterations each instance needs to reach the tolerance, of 1024 MOT17 frame pairs (\texttt{mot-b16} pool of \cref{fig:a3}); the last bin holds instances at the cap of 500.}
\label{fig:teaser}
\end{figure}

\drain{} instead makes each problem of the batch, which we call an \emph{instance}, the unit of GPU execution. Wherever an instance fits on chip, its whole state stays in the registers of one GPU program (a thread block) from the first update to the last. We call such an instance \emph{resident}. When a resident instance converges, its program exits and the GPU starts a pending instance in the freed slot. We call this \emph{release}. Most of the speedup comes from keeping instances on chip. Release adds to it when a batch needs several waves of programs. We make three contributions.
\begin{itemize}
\item \textbf{Instance-resident execution for both stopping rules} (\cref{sec:method}). Tolerance solves release each instance by program exit, for dense instances of up to 64 points per side, point clouds in $d\le4$ with up to 512, and grids. Fixed-iteration solves pack several instances into one program and run the forward map and its exact backward in one launch each, without host synchronization. Instances with one short and one very long side run in one launch when they fit on the device.
\item \textbf{Per-instance release with a runtime model} (\cref{sec:method,sec:analysis}). Instances too large for one program run on an indexed path that removes converged instances from every later launch. A list-scheduling bound explains when release shortens a batch whose programs run in several waves, the groups of programs that the GPU runs at once.
\item \textbf{An in-loop benchmark} (\cref{sec:experiments}) built from the users above: label assignment captured from COCO with a trained detector, mHC projections, patch matching, tracking, and domain decomposition. Each is compared with the application's own implementation where one exists and with general solvers under common accuracy gates, and we report the solver's share of a training step.
\end{itemize}

On label assignment with two images per call, one GPU's share of OTA's 8-GPU step, \drain{} is \resOTApgOfficial$\times$ faster than OTA's own implementation and \resOTApgPadded$\times$ faster than a padded batch of the same updates. It cuts assignment from \resOTApgShareOfficial\% to \resOTApgShareDrain\% of a detector's training step without changing any assignment, and makes GeoTransformer's patch matching \resPatchTrain$\times$ faster in a training step (\cref{tab:train}). On the inference workloads it is the fastest of all compared methods, \resMainSpeedupRange$\times$ faster than the fastest baseline, OTT-JAX, and \resMainPOTRange$\times$ faster than POT (\cref{tab:main}). The gains hold on an H100 and an L40S (\cref{app:gpus,app:gpus-train}), and release adds \resRetireWallRange$\times$ on multi-wave batches (\cref{sec:mechanism}).

\section{Problem}
\label{sec:problem}

\paragraph{Batches of independent problems.}
A batch is a set $\{q_b\}_{b=1}^{B}$ of balanced entropic OT problems $q_b=(C_b,a_b,b_b,\eps_b)$ with a cost $C_b\in\R^{N\times M}$, given either densely or as $C_{b,ij}=\|x_{b,i}-y_{b,j}\|^2$ for supports $X_b\in\R^{N\times d}$, $Y_b\in\R^{M\times d}$, and nonnegative marginals of equal mass. We call each $q_b$ an \emph{instance} of the batch. Instance $b$ is
\begin{equation}
\min_{P\ge 0}\ \langle C_b,P\rangle+\eps_b\sum_{ij}P_{ij}(\log P_{ij}-1)
\quad\text{s.t.}\quad P\ones=a_b,\quad P^\top\ones=b_b ,
\label{eq:eot}
\end{equation}
and its solution $P_b$ is the plan of instance $b$. Instances share no variables, and $B$ counts instances, not points. Instances may share a support, as domain-decomposition tiles do, but not solver state.

\paragraph{Updates.}
We use log-domain Sinkhorn with dual potentials $(f,g)$ and the plan $P_{ij}=\exp((f_i+g_j-C_{ij})/\eps)$. One iteration is the row update followed by the column update. We call each of the two a \emph{half-step}:
\begin{equation}
f_i \leftarrow \eps\bigl(\log a_i-\LSE_j\bigl((g_j-C_{ij})/\eps\bigr)\bigr),\qquad
g_j \leftarrow \eps\bigl(\log b_j-\LSE_i\bigl((f_i-C_{ij})/\eps\bigr)\bigr),
\label{eq:updates}
\end{equation}
starting from $f=g=0$. After the column update, $P^\top\ones=b$ exactly, but the row marginal does not.

\paragraph{Stopping at a tolerance.}
Every $q$ iterations, instance $b$ checks its two-sided residual $r_b^{(t)}=\max\{\|P\ones-a_b\|_1,\|P^\top\ones-b_b\|_1\}$ and stops at the first check where $r_b^{(t)}\le\tau$, after its \emph{first-passage depth} $D_b$ of iterations, or returns as \emph{capped} after $K_{\max}$ iterations.

\paragraph{Running exactly $K$ iterations.}
When Sinkhorn is a layer of a trained model, the layer is the $K$-step map itself: mHC projects each residual mixing matrix with 20 iterations \citep{xie2025mhc}, OTA assigns labels with 50 \citep{ge2021ota}, and GeoTransformer matches patches with 100 \citep{qin2022geotransformer}. With this stopping rule the solver runs exactly $K$ iterations of \eqref{eq:updates} from $f=g=0$ and returns that iterate. When gradients are requested, it also returns the exact vector--Jacobian product of the $K$-step map with respect to $C_b$, $a_b$, $b_b$, and $\eps_b$, obtained by unrolled differentiation rather than as an implicit gradient at the optimum. Under either rule, for each instance in input order, the solver returns $(f_b,g_b)$ or the plan $P_b$, together with the residuals, $D_b$, and the status.

\paragraph{Cost of fixed-width execution.}
A fixed-width batch runs every instance until the slowest one stops. It performs $B\Dmax$ instance-updates, one per iteration of one instance, where $\Dmax=\max_b D_b$, while stopping each instance at its own depth needs only $\sum_b D_b$. The ratio
\begin{equation}
\Rwork=\frac{B\,\Dmax}{\sum_b D_b}\ \ge 1
\label{eq:rwork}
\end{equation}
depends only on the workload. Across our workloads it ranges over \resRworkRange. \Cref{sec:analysis} relates it to runtime.

\section{Instance-resident execution}
\label{sec:method}

\drain{} makes the instance the unit of GPU execution, not the batch, which general solvers advance with one launch per half-step until its slowest instance stops. Wherever an instance fits, one GPU program, Triton's name for a thread block, runs the instance's whole solve: the instance's state stays on chip from the first update to the last, and the host issues one launch for the batch. Each program reads and writes only its own instance's data. The execution path depends only on the stopping rule and the instance shape (\cref{fig:method}). With a tolerance, dense instances with $N,M\le64$ and squared-Euclidean instances with $d\le4$ and $N,M\le512$ run as resident solves with release (\cref{sec:resident}), and regular grids with at most 16 points per axis can request a separable resident solve (\cref{sec:separable}). With a fixed iteration count, dense instances with $N,M\le128$ run as packed resident forward and backward passes (\cref{sec:fixed}), and short-by-long instances, such as the images of a label-assignment step, run in chunks along the long side (\cref{sec:wide}). All other instances use indexed execution over unfinished instances (\cref{sec:indexed}). All paths return one per-instance result type, and all kernels are written in Triton.

\begin{figure}[t]
\centering
\begin{tikzpicture}[node distance=0.12cm and 0.35cm, scale=0.92, every node/.append style={transform shape}]
  \node[font=\scriptsize\bfseries, anchor=west] (ta) at (0,0) {(a) Resident solve with release};
  \node[sbox, below=0.25cm of ta.south west, anchor=north west, text width=2.6cm] (load) {load $C_b$ or $X_b,Y_b$,\\$a_b$, $b_b$ on chip};
  \node[sbox, below=of load, text width=2.6cm] (upd) {$q$ full updates\\\eqref{eq:updates}};
  \node[sbox, below=of upd, text width=2.6cm] (chk) {rebuild $P_b$ on chip;\\row and column $\ell_1$};
  \node[box, below=of chk, text width=2.6cm, fill=draingreen!15] (out) {if $r_b\le\tau$ or $t=K_{\max}$:\\write $f_b,g_b,r_b,D_b$;\\exit};
  \draw[arr] (load) -- (upd);
  \draw[arr] (upd) -- (chk);
  \draw[arr] (chk) -- (out);
  \draw[arr] (chk.west) -- ++(-0.3,0) |- node[pos=0.25, left, font=\tiny, align=right]{not\\passed} (upd.west);
  \node[draw=drainblue, dashed, rounded corners=3pt, fit=(load)(upd)(chk)(out), inner sep=4pt, label={[font=\tiny, text=drainblue]below:{one GPU program per instance}}] (prog) {};
  \node[box, right=0.45cm of upd, text width=2.3cm, fill=white] (sched) {hardware block\\scheduler};
  \node[box, above=0.3cm of sched, text width=2.3cm, fill=draingray!15] (pend) {pending programs\\$q_{S+1},\dots,q_B$};
  \node[box, below=0.3cm of sched, text width=2.3cm, fill=draingray!15] (slots) {$S$ concurrent\\program slots};
  \draw[arr] (pend) -- (sched);
  \draw[arr] (sched) -- (slots);
  \draw[arr, draingreen] (out.east) -| node[pos=0.8, right, font=\tiny, text=draingreen]{slot freed} (slots.south);

  \begin{scope}[xshift=7.2cm]
  \node[font=\scriptsize\bfseries, anchor=west] (tb) at (-0.2,0) {(b) Indexed active execution};
  \node[sbox, below=0.25cm of tb.south west, anchor=north west, text width=3.0cm] (rowu) {row update on $\Aset$\\(streamed $\LSE$, $C$ on the fly)};
  \node[sbox, below=of rowu, text width=3.0cm] (colu) {column update on $\Aset$};
  \node[box, below=of colu, text width=3.0cm, fill=drainorange!15] (scr) {check step: next row update\\$f^{\mathrm{bal}}$ with the screen\\$s_b$ of \eqref{eq:screen} in its epilogue};
  \node[box, below=of scr, text width=3.0cm, fill=drainorange!15] (ver) {recompute both residuals of the\\\emph{current} $(f,g)$ of $\{b: s_b\le\tau\}$};
  \node[box, below=of ver, text width=3.0cm, fill=draingreen!15] (ret) {save accepted state;\\$\Aset\leftarrow\Aset\setminus\text{accepted}$;\\remaining ones reuse $f^{\mathrm{bal}}$};
  \draw[arr] (rowu) -- (colu);
  \draw[arr] (colu) -- node[right, font=\tiny]{every $q$} (scr);
  \draw[arr] (scr) -- (ver);
  \draw[arr] (ver) -- (ret);
  \draw[arr] (ret.east) -- ++(0.3,0) |- (colu.east);
  \draw[arr] (colu.west) -- ++(-0.3,0) |- (rowu.west);
  \node[font=\tiny, text=black!70, align=left, anchor=north west] at ($(ret.south west)+(0,-0.12)$) {$X,Y$, norms, and outputs stay in their\\original slots; only $\Aset$ shrinks.};
  \end{scope}
\end{tikzpicture}
\caption{\textbf{Two realizations of per-instance release.} (a)~Instances that fit on chip run the whole loop in one program. Release ends the program and the hardware scheduler fills the slot. The fixed-iteration path (\cref{sec:fixed}) runs the same program without the check, several instances per program. (b)~Larger instances iterate over an index set $\Aset$. Instances that pass the screen are checked on their current iterate, and those that pass leave $\Aset$.}
\label{fig:method}
\end{figure}

\subsection{Resident solves with release}
\label{sec:resident}

When both supports have at most 64 points, the cost tile of one instance, its marginals, and its potentials fit in the registers of one program. \drain{} assigns one program to each instance and runs the complete solve there. Each program repeats four steps. It performs $q$ iterations of \eqref{eq:updates} on chip, rebuilds the plan tile from the current potentials, reduces the tile's rows and columns to the two marginal $\ell_1$ residuals, and stops if both are at most $\tau$ or the cap is reached. On release the program writes $f_b$, $g_b$, both residuals, and $D_b$, and exits.

Release is program termination. The batch is one launch of $B$ programs, and the GPU block scheduler starts a pending program whenever a running one exits, so release needs no host decision, index vector, or synchronization until the batch is done.

\paragraph{Larger low-dimensional supports.}
For squared-Euclidean costs with $d\le4$ and up to 512 points per side, the cost tile no longer fits, but the coordinates, squared norms, marginals, and potentials of one instance do. The program keeps these on chip and recomputes distances tile by tile inside each half-step, so neither $C$ nor $P$ is materialized. The check and release are unchanged. Whether this path or the indexed one is faster depends on the device. On first use, \drain{} therefore runs the indexed path and three resident tilings on the calling batch, admits a tiling only if its registers do not spill to memory and its per-instance stopping status matches the indexed path's, and keeps the fastest tiling per device, batch size, and shape for the rest of the process.

\subsection{Separable grids}
\label{sec:separable}

Regular grids, such as the cells of a domain decomposition, have Cartesian supports, and the squared Euclidean cost splits as $C_{(i_1i_2)(j_1j_2)}=C^{(1)}_{i_1j_1}+C^{(2)}_{i_2j_2}$. The log-sum-exp of a half-step then factorizes into two one-axis reductions,
\begin{equation}
\LSE_{j_1,j_2}\bigl(h_{j_1j_2}-C^{(1)}_{i_1j_1}/\eps-C^{(2)}_{i_2j_2}/\eps\bigr)
=\LSE_{j_1}\Bigl(-C^{(1)}_{i_1j_1}/\eps+\LSE_{j_2}\bigl(h_{j_1j_2}-C^{(2)}_{i_2j_2}/\eps\bigr)\Bigr),
\label{eq:separable}
\end{equation}
which lowers the work of a half-step on an $s\times s$ grid from $s^4$ to $2s^3$. GPU domain-decomposition solvers use this reduction \citep{logsinkhorngpu,medina2025flow}, and \drain{} runs it inside a resident program. With at most 16 points per axis, the axis costs, the potentials, and the intermediate $s\times s$ array fit in one program, so the resident loop above, including the check and release, runs unchanged on grids of up to 256 points per side. The option is explicit: the caller supplies the grid shapes.

\subsection{Fixed-iteration forward and backward}
\label{sec:fixed}

A training-time layer calls the solver in every step, often on very many tiny instances (mHC projects one $4\times4$ matrix per token and layer), and needs gradients. There is nothing to release, but the costs of batch-organized execution remain: one launch per half-step, and an autograd graph that stores the intermediate plans of all $K$ iterations. \drain{} runs the whole $K$-step map of an instance in one resident launch and its vector--Jacobian product in one more. mHC already runs its $4\times4$ projection in one fused kernel and recomputes it on chip in a custom backward \citep{xie2025mhc}. This path extends that design to per-instance costs, marginals, and $\eps$, dense instances of up to 128 points per side, gradients with respect to all inputs, and several instances per program.

\paragraph{Packing.}
A $4\times4$ instance would leave almost all threads of a program idle. One program therefore holds several instances along a leading axis: 32 for padded tiles of at most 64 elements (the mHC shapes) and $\max(1,\lfloor 2048/(RC)\rfloor)$ for larger tiles of $R\times C$ elements, with masks for padding and for the partly filled last program. Reductions run only along each instance's own rows and columns, so packing changes no instance's iterates.

\paragraph{Forward.}
Each program loads the cost tiles, marginals, and $\eps_b$ of its instances and runs the $K$ iterations on chip in base 2, with potentials and costs prescaled by $\log_2 e/\eps_b$ so that no half-step divides. The program writes the final potentials and, on request, the plan and both marginal residuals from the same tile. For gradients, it also writes the potential history $(f^{(t)},g^{(t)})_{t\le K}$, $K(N+M)$ numbers per instance instead of $K$ plans of $NM$ (\cref{app:capture}).

\paragraph{Backward.}
A second resident launch reloads the cost tiles, turns the adjoint of the plan into adjoints of $C_b$, $\eps_b$, and the final potentials, and walks the $2K$ half-steps in reverse. The normalized weights $\exp((g_j-C_{ij})/\eps-\LSE)$ of each half-step are rebuilt on chip from the saved potentials and consumed immediately, and the adjoints of $f$ and $g$ accumulate into gradients with respect to $C_b$, $a_b$, $b_b$, and $\eps_b$. The result is the exact vector--Jacobian product of the $K$-step map.

\paragraph{Host control.}
This path issues a fixed, data-independent number of launches. Without input value checks, the forward and backward passes need no host synchronization and fit in a CUDA graph.

\subsection{Short-by-long instances}
\label{sec:wide}

We call an instance short-by-long when one side has at most 128 points and the other more than 4096. Label assignment is the example: each image couples its few ground-truth boxes with tens of thousands of anchors, the candidate locations of a detector (\cref{sec:setup}). A whole instance does not fit on chip, but its short side does. \drain{} splits the long side into chunks, one program each. A program keeps the short-side potential in registers, updates the long-side potential of its chunk by a reduction over the short side, and leaves one partial log-sum-exp per short index. The next short-side update starts by combining these partials. Programs visit rows in blocks over each instance's actual short side, so images with few boxes do no work on padded rows. When a call needs at most one program per streaming multiprocessor (SM), with chunks of 1024 or 2048 entries, a single cooperative launch, for which the driver guarantees that all programs run at the same time, runs all $K$ iterations with a counter between them. Otherwise, each iteration is one launch that reads the cost once. Because a box with no anchor inside would carry a potential near OTA's outside-box cost of $10^6$, which FP32 cannot resolve, \drain{} first subtracts each short-side entry's minimum cost and starts that entry's potential at minus the shift. In exact arithmetic this runs the same $K$-step map from $f=g=0$ whichever side is updated first. Costs are stored short side first, and either side may be updated first, so a transposed view reproduces OTA's order.

\subsection{Indexed active execution}
\label{sec:indexed}

Larger instances do not fit on chip, so each half-step is a streamed kernel: for a tile of rows, the kernel computes squared distances from centered coordinates and cached norms, accumulates an online log-sum-exp over column tiles, and writes the updated potential, without materializing $C$ or $P$ \citep{feydy2019sinkhorn,ye2026flashsinkhorn}. Every launch takes an index vector $\Aset$ of unfinished instances, and the programs of the launch cover only these $|\Aset|$ instances. State and outputs keep their batch slots.

\paragraph{Screen in the next update.}
At a check, the column marginal is exact, so the two-sided residual reduces to the row residual. Let $f^{\mathrm{bal}}$ be the next row update computed from the current $g$. The row marginal of the current plan is $(P\ones)_i=a_i\exp((f_i-f^{\mathrm{bal}}_i)/\eps)$, hence
\begin{equation}
\|P\ones-a\|_1=\sum_i a_i\,\bigl|\expmone\bigl((f_i-f^{\mathrm{bal}}_i)/\eps\bigr)\bigr| .
\label{eq:screen}
\end{equation}
The identity is known \citep{ye2026flashsinkhorn}. We use it as a screen that costs no extra pass, because the final stage of the row-update kernel that produces $f^{\mathrm{bal}}$ also accumulates \eqref{eq:screen}. Instances with screen value at most $\tau$ go to a separate streamed kernel that recomputes both marginals of the current iterate. Instances that pass copy $f$, $g$, and residuals to the output and leave $\Aset$. For instances that stay active, $f^{\mathrm{bal}}$ is exactly their next row update, so the screen adds only that final stage. Each check costs one device-to-host transfer of the pass flags and one upload of the new index vector. At the cap, the remaining instances are checked once more and marked capped (\cref{alg:drain}, \cref{app:impl}).

\section{Work-to-runtime analysis}
\label{sec:analysis}

Release removes work. This section states when it also saves time, in exact arithmetic, with proofs in \cref{app:proofs}. With $\Aset_t=\{b: D_b\ge t\}$ the set of instances active at iteration $t$, $\sum_{t\le\Dmax}|\Aset_t|=\sum_b D_b$, so release removes $B\Dmax-\sum_b D_b$ instance-updates \refstepcounter{equation}(\theequation)\label{eq:work}.

\paragraph{Resident path: waves of programs.}
Suppose the GPU runs at most $S$ programs at once, one per slot, the program of instance $b$ takes $\theta D_b+\sigma$ time ($\theta$ per iteration and a fixed cost $\sigma$), and programs do not slow each other down. The block scheduler is work-conserving: it starts a pending program whenever a slot is free. The programs that run at the same time form a \emph{wave}.

\begin{proposition}[Resident makespan]
\label{prop:makespan}
Under these assumptions the makespan $T_{\mathrm{active}}$ of execution with release satisfies
\begin{equation}
\max\Bigl(\tfrac1S\textstyle\sum_b(\theta D_b+\sigma),\ \theta\Dmax+\sigma\Bigr)
\ \le\ T_{\mathrm{active}}\ \le\
\tfrac1S\textstyle\sum_b(\theta D_b+\sigma)+\bigl(1-\tfrac1S\bigr)(\theta\Dmax+\sigma),
\label{eq:makespan}
\end{equation}
while a fixed-width batch that runs every program to $\Dmax$ takes $T_{\mathrm{static}}=\lceil B/S\rceil(\theta\Dmax+\sigma)$.
\end{proposition}

The upper bound is Graham's bound for list scheduling \citep{graham1969bounds}. With $\sigma=0$, the ratio $T_{\mathrm{static}}/T_{\mathrm{active}}$ lies in an interval computable from the depths alone, which \cref{sec:mechanism} compares with measurement. When the batch fits in one wave ($B\le S$), $T_{\mathrm{static}}$ equals the lower bound, so release cannot help. When the batch spans many waves, $T_{\mathrm{static}}/T_{\mathrm{active}}\to(\theta\Dmax+\sigma)/(\theta\bar D+\sigma)$ for the mean depth $\bar D$, and this limit approaches $\Rwork$ of \eqref{eq:rwork} as $\sigma\to0$.

\paragraph{Indexed path: widths and control.}
An indexed half-step over $a$ active instances with $n$ rows launches $a\lceil n/b_m\rceil$ programs for row tile size $b_m$, which the GPU runs in waves of at most $Q$ programs, so an iteration costs about $g_n(\lceil a\lceil n/b_m\rceil/Q\rceil)$ for a function $g_n$ that increases with the number of waves. Shrinking $\Aset$ therefore saves time only when it removes a wave, and every check adds host control (\cref{sec:indexed}). With few row tiles per instance, the removed work \eqref{eq:work} can be large while the removed time is small. The resident path targets these small instances.

\section{Experiments}
\label{sec:experiments}

\subsection{Setup}
\label{sec:setup}

\paragraph{Inference workloads.}
Seven workloads, each a batch of 64 to 3771 instances with 16 to 1024 points per side (\cref{app:workloads}). \emph{MOT17} \citep{milan2016mot16}: detections of consecutive frames are matched in a four-dimensional box feature space. One call holds all frame pairs of the sequences (offline association), grouped by detection count into \texttt{mot-b16}, \texttt{mot-b32}, and \texttt{mot-b64}, or cut into windows of 1024 detections (\texttt{mot-win}). \emph{Domain decomposition} (DD): images of the DomainDecomposition reference \citep{bonafini2021domain} split into $4\times4$, $8\times8$, and $16\times16$ tiles (\texttt{dd-t4}, \texttt{dd-t8}, \texttt{dd-t16}), a controlled proxy for DD cell problems.

\begin{table}[t]
\centering
\caption{\textbf{Training layers, fixed iterations} (other method's time over \drain's, paired per block. Bold: other method faster. $^\dagger$Fails the same-map gate. $^*$Computes a different quantity: TileLang-mHC adds $10^{-6}$ to every entry after its first normalization, and its backward is an implicit gradient at its output, not the gradient of the 20-step map). (a)~OTA, 20 COCO steps, 50 iterations, 2 or 16 images per call. Share: assignment share of the step (median, \cref{app:capture}), official $\to$ \drain. (b)~mHC, $4\times4$ inputs captured from a trained GPT, 20 iterations, all replayed from CUDA graphs. ms: \drain{} time per call. Compile: the PyTorch loop under \texttt{torch.compile}. Triton: triton-sinkhorn. (c)~GeoTransformer, 32 3DMatch pairs, 100 iterations. In model: OT time inside the model at inference or in a training step, median over pairs. Alone: forward calls. Share: OT share of the step.}
\label{tab:train}
\small
\setlength{\tabcolsep}{3pt}
\begin{tabular*}{\linewidth}{@{\extracolsep{\fill}}lrrrrrrr}
\toprule
\multicolumn{8}{l}{(a) Label assignment} \\
Call & images & \drain{} (ms) & OTA official & padded & POT & OTT & share (\%) \\
\midrule
one GPU's share & 2 & 1.88 & 13.0 & 6.31 & 4.59 & 1.05 & 16.3\,$\to$\,4.7 \\
whole step & 16 & 7.95 & 24.4 & 6.29 & 4.85$^\dagger$ & 2.15 & 15.4\,$\to$\,4.4 \\
 
\bottomrule
\end{tabular*}

\vspace{3pt}
\begin{tabular*}{\linewidth}{@{\extracolsep{\fill}}lrrrrrrrrr}
\toprule
\multicolumn{10}{l}{(b) mHC projections} \\
& & \multicolumn{4}{c}{forward} & \multicolumn{4}{c}{forward and backward} \\
\cmidrule(lr){3-6}\cmidrule(lr){7-10}
Input & $B$ & ms & compile & triton & TileLang$^*$ & ms & compile & triton & TileLang$^*$ \\
\midrule
hard layer & 4K & 0.057 & 2.94 & 7.91 & \textbf{0.68} & 0.244 & 2.66 & 2.56 & \textbf{0.22}$^\dagger$ \\
medium layer & 4K & 0.057 & 2.95 & 7.86 & \textbf{0.68} & 0.246 & 2.63 & 2.54 & \textbf{0.21}$^\dagger$ \\
easy layer & 4K & 0.057 & 2.91 & 7.88 & \textbf{0.68} & 0.243 & 2.60 & 2.57 & \textbf{0.22}$^\dagger$ \\
16 layers & 64K & 0.151 & 1.83 & 41.1 & 1.02$^\dagger$ & 0.453 & 2.60 & 14.3 & \textbf{0.48}$^\dagger$ \\
 
\bottomrule
\end{tabular*}

\vspace{3pt}
\begin{tabular*}{\linewidth}{@{\extracolsep{\fill}}lrrrrr}
\toprule
\multicolumn{6}{l}{(c) Patch matching} \\
Setting & \drain{} (ms) & official & OTT & POT & share (\%) \\
\midrule
alone, forward & 1.34 & 10.0 & 2.80 & 11.9 & -- \\
in model, inference & 1.43 & 9.94 & -- & -- & 17.2\,$\to$\,2.1 \\
in model, training step & 2.65 & 10.4 & -- & -- & 16.4\,$\to$\,1.9 \\
 
\bottomrule
\end{tabular*}
\end{table}

\paragraph{Training layers.}
\emph{OTA label assignment} is captured from real data: 320 COCO val2017 images in 20 steps of 16, with OTA's official cost, supplies, and demands \citep{ge2021ota} built from a COCO-trained FCOS \citep{tian2019fcos} at its training resolution (\cref{app:capture}). Each step has 16 instances of 2 to 54 rows, padded to 27K to 36K anchors, with 50 iterations and unnormalized marginals. \emph{mHC}: logits and upstream gradients captured from a trained mHC model, and synthetic logits $\sim\mathcal N(0,\sigma^2)$, $\sigma\in\{1,10\}$ \citep{wang2026birkhoff}, with $n\in\{4,8\}$ and up to 1M instances per call, with 20 iterations. \emph{Patches}: GeoTransformer's fine matching captured from its released model on 32 3DMatch test pairs \citep{qin2022geotransformer}. Each pair gives 256 instances of $64\times64$ scores plus a dustbin row and column that absorb unmatched points, with a median of \resPatchValid{} valid points per side, solved with 100 iterations. \emph{Crowd counting}: DM-Count's OT loss on 80 crops of 20 annotated crowd images, 100 iterations per image (\cref{app:e2e}).

\paragraph{Baselines.}
Batched log-domain \emph{POT} \resPOTVersion{} \citep{flamary2021pot}, \emph{GeomLoss} \resGeomLossVersion{} \citep{feydy2019sinkhorn}, and \emph{OTT-JAX} \resOTTVersion{} \citep{cuturi2022ott} as \texttt{jit(vmap(Sinkhorn))} with highest matmul precision, plus each application's own implementation: the PyTorch reference loop of mHC \citep{xie2025mhc}, also compiled with \texttt{torch.compile} \citep{ansel2024pytorch2}, triton-sinkhorn, a public fused Triton kernel for the mHC layer \citep{tritonsinkhorn}, TileLang-mHC, the Sinkhorn kernel of TileLang's mHC example as benchmarked by mHC-proj \citep{wang2025tilelang,wang2026birkhoff}, and the official modules of OTA \citep{ge2021ota} and GeoTransformer \citep{qin2022geotransformer}.

\paragraph{Accuracy gates.}
All methods request $\tau=10^{-4}$, and one FP32 routine, the consumer, rebuilds every plan from the returned potentials. A run counts only if every instance has a two-sided residual of at most $2\cdot10^{-3}$, which admits capped instances. Fixed-iteration methods must reproduce the FP64 $K$-step map (the same-map gate): every assignment (OTA), forward outputs within $10^{-4}$ and input gradients within $10^{-3}$ relative (mHC), or forward outputs within $10^{-4}$ (patches).

\paragraph{Timing.}
Each timed call runs from device-resident inputs to the returned output (for \cref{tab:main}, through the consumer), including input checks and preparation. Each method runs in its own process after warm-up, with the order rotated across \resOuterBlocks{} outer blocks. We report medians over blocks and ratios paired per block with bootstrap intervals, on one NVIDIA \resGPU{} with PyTorch \resTorch{} and CUDA \resCUDA, in FP32 without TF32. \Cref{app:gpus,app:gpus-train} repeat the workloads on an H100 and an L40S. The released code holds the benchmark scripts and baseline configurations.

\subsection{Training layers}
\label{sec:inloop}

\paragraph{Label assignment.}
OTA's official configuration trains with 16 images per step on 8 GPUs \citep{ge2021ota}, so each GPU assigns 2 images per call. On this call \drain{} is the fastest method at \resOTApgDrainMs~ms, \resOTApgOfficial$\times$ faster than OTA's per-image loop (\cref{tab:train}a). OTT-JAX nearly matches it (\resOTApgOTT$\times$). In a one-GPU training step, with the detector's forward, loss, and backward timed apart from the assignment, assignment falls from \resOTApgShareOfficial\% to \resOTApgShareDrain\% of the step and its Sinkhorn part from \resOTApgSinkShareOfficial\% to \resOTApgSinkShareDrain\%, a \resOTApgStep$\times$ faster step with no assignment changed.

\paragraph{mHC projections.}
With the layer replayed from CUDA graphs (\cref{tab:train}b), \drain{} is \resMHCCompile$\times$ faster than the PyTorch loop under \texttt{torch.compile} and \resMHCTriton$\times$ faster than triton-sinkhorn on single sublayers, and 1.83--2.60$\times$ and 14.3--41.1$\times$ with all 16 sublayers per call. TileLang-mHC is \resMHCTileF$\times$ faster in the forward pass on single sublayers and 2.08--4.71$\times$ with the backward, but its backward differs from the gradient of the 20-step map by \resMHCTileGrad\%, and over 2{,}500 training steps of the GPT it ends at a held-out loss of \longTileRange{}, against \longExactRange{} for the official loop and \drain{} (three seeds each, \cref{app:long}). Against the official loop, a whole GPT training step is \resEtoEMHCDrain$\times$ faster with \drain{}, \resEtoEMHCCompile$\times$ with \texttt{torch.compile}, and \resEtoEMHCTile$\times$ with TileLang-mHC (\cref{app:e2e}).

\paragraph{Crowd counting.}
DM-Count \citep{wang2020dmcount} trains a density map with an OT loss, 100 iterations per image between the $64\times64$ predicted density and the annotated heads, one image after another. On 80 crops of dense crowds, \drain{} makes this loss \resDMOT$\times$ and the training step \resDMStep$\times$ faster (loss share \resDMShareOfficial\% to \resDMShareDrain\%). The official FP32 solver drops the mass of cells far from every head, so the gradients differ (\cref{app:e2e}).

\paragraph{Patch matching.}
GeoTransformer pads every patch to 64 points, while \drain{} solves each instance at its valid size. Inside the model (\cref{tab:train}c), OT becomes \resPatchInfer$\times$ faster at inference and \resPatchTrain$\times$ faster in a training step, a \resPatchTrainStep$\times$ faster step, and the same \resPatchRegistered{} of \resPatchPairs{} pairs register with either module. On forward calls, where POT and OTT-JAX get the padded $65\times65$ problems, \drain{} is also fastest, \resPatchPeers$\times$ faster than OTT-JAX and POT (\cref{app:capture}). On four pairs, peak memory above the inputs in forward and backward falls from \resMemPatchOfficial~MB (official) to \resMemPatchDrain~MB (\cref{app:capture}). Replacing only the solver bounds the step speedup by $1/(1-s)$ for the solver's share $s$ of the step: 1.15$\times$ for OTA's Sinkhorn share of \resOTApgSinkShareOfficial\% and 1.20$\times$ for GeoTransformer's OT share of \resPatchTrainShareOfficial\%, against the measured \resOTApgStep$\times$ and \resPatchTrainStep$\times$.

\subsection{Inference workloads}
\label{sec:main}

\begin{table}[t]
\centering
\caption{\textbf{Time per call} (ms, median of \resOuterBlocks{} blocks, with validation, preparation, and consumer). Speedup: OTT-JAX time over \drain{}, paired per block, with 95\% bootstrap CI. All outputs pass the gate of \cref{sec:setup}. GeomL.: GeomLoss, which takes no tolerance and runs a fixed annealing schedule. OTT: OTT-JAX. $\le\tau$: share of instances within $\tau=10^{-4}$ in FP64, rounded down. $\Rwork$: \eqref{eq:rwork}.}
\label{tab:main}
\small
\setlength{\tabcolsep}{2.4pt}
\begin{tabular}{lrrrrrlrr}
\toprule
& & & & & & & \multicolumn{2}{c}{$\le\tau$ (\%)} \\
\cmidrule(lr){8-9}
Workload & $\Rwork$ & \drain & POT & GeomL. & OTT & Speedup [CI] & \drain & OTT \\
\midrule
\multicolumn{9}{l}{\emph{Converged within the default cap}}\\
\texttt{dd-t8} & 1.19 & 4.68 & 40.7 & 1220 & 23.9 & 5.12\,[5.09, 5.15] & 99 & 97 \\
\texttt{mot-b64} & 1.40 & 2.68 & 25.7 & 1224 & 9.81 & 3.68\,[3.62, 3.73] & 97 & 95 \\
\texttt{dd-t16} & 1.19 & 20.6 & 109 & 580 & 56.3 & 2.73\,[2.72, 2.74] & 100 & 99 \\
\texttt{mot-win} & 2.59 & 101 & 582 & 10031 & 306 & 3.04\,[3.02, 3.05] & 98 & 98 \\
\midrule
\multicolumn{9}{l}{\emph{Budget-limited at the default cap}}\\
\texttt{dd-t4} & 2.43 & 9.58 & 327 & 9730 & 101 & 10.53\,[10.49, 10.62] & 84 & 84 \\
\texttt{mot-b16} & 3.70 & 5.85 & 85.1 & 3714 & 62.9 & 10.75\,[10.67, 10.92] & 85 & 88 \\
\texttt{mot-b32} & 4.70 & 4.65 & 83.0 & 3714 & 42.6 & 9.17\,[9.01, 9.28] & 98 & 97 \\
 
\bottomrule
\end{tabular}
\end{table}

OTT-JAX is the fastest baseline on all seven workloads (\cref{tab:main}), and \drain{} is \resMainSpeedupRange$\times$ faster than OTT-JAX and \resMainPOTRange$\times$ faster than POT. POT stops the whole batch on one error. Its median residuals are \resPOTOverSolve$\times$ smaller than those of \drain{} and OTT-JAX (\cref{app:audit}). On three workloads some instances reach the default cap under every method. With a fivefold cap at $\tau=10^{-4}$, \drain{} is \resTolOTTddfour$\times$ and \resTolOTTmotsixteen$\times$ faster than OTT-JAX on \texttt{dd-t4} and \texttt{mot-b16}, with median residuals within 1\% and median depths within one check interval of OTT-JAX's (\cref{app:more}). The POT ratio also grows with the cap (\cref{fig:cap}). On grids, the separable option beats LogSinkhornGPU's cell solver by \resSepRange$\times$ where both converge (\cref{app:more}), and on replayed flow-matching OT batches \drain{} is \resReplayOTT$\times$ faster than OTT-JAX (\cref{app:replay}).

\subsection{Sources of the speedup}
\label{sec:mechanism}

\begin{table}[t]
\centering
\caption{\textbf{Sources of the speedup} per inference workload: paired medians per block, range over an A100, an H100, and an L40S (\cref{app:gpus}). Indexed: OTT-JAX time over our indexed path with every instance run to $\Dmax$ (below 1: our indexed path is slower). On chip: the indexed path over the same updates as resident solves, run to $\Dmax$ with the same checks. Release: the resident (or, on \texttt{mot-win}, indexed) kernel run to $\Dmax$ over the same kernel with release. \texttt{dd-t16} runs on the low-dimensional resident path, which has no variant without release, so its two factors are given together. Total: OTT-JAX over \drain{}.}
\label{tab:source}
\small
\begin{tabular}{llrrrr}
\toprule
Workload & Path & Indexed & On chip & Release & Total \\
\midrule
\texttt{dd-t4} & resident & 0.11--0.14 & 62.2--80.0 & 1.29--1.43 & 11.1--13.4 \\
\texttt{dd-t8} & resident & 0.23--0.32 & 15.0--18.6 & 1.01--1.09 & 4.32--6.33 \\
\texttt{dd-t16} & resident, $d\le4$ & 0.51--0.81 & \multicolumn{2}{c}{3.08--4.33 (both)} & 1.70--2.64 \\
\texttt{mot-b16} & resident & 0.16--0.26 & 29.1--31.4 & 1.23--1.41 & 5.56--10.6 \\
\texttt{mot-b32} & resident & 0.15--0.22 & 27.2--34.1 & 1.33--1.55 & 6.59--10.0 \\
\texttt{mot-b64} & resident & 0.14--0.18 & 17.1--17.5 & 1.00--1.03 & 2.36--3.12 \\
\texttt{mot-win} & indexed & 1.46--5.18 & -- & 1.67--2.12 & 2.62--10.1 \\
 
\bottomrule
\end{tabular}
\end{table}

Most of the gain comes from on-chip execution (\cref{tab:source}). Our indexed path with every instance run to $\Dmax$ is slower than OTT-JAX on the resident workloads, and running the same updates on chip makes them \gpuOnChipRange$\times$ faster. Release adds \gpuReleaseResRange$\times$ on the resident path, most on batches that span several waves. Holding each kernel fixed and removing only release (\cref{fig:mechanism}), release cuts wall time by \resRetireWallRange$\times$ and solver time by \resRetireSolveRange$\times$ on the four resident batches spanning two or three A100 waves, against the same kernel run to $\Dmax$ with the same checks. The solver-time ratios lie within the interval of \cref{prop:makespan} except on \texttt{mot-b32} (2.01 vs.\ 2.00). \texttt{mot-b64} fits in one wave yet gains \resRetireOneWaveWall$\times$, which the proposition does not explain. On the indexed path (\texttt{dd-t16}, \texttt{mot-win}), masking finished instances costs \resMaskedDDsixteen$\times$ and \resMaskedWin$\times$, and running all to $\Dmax$ \resStaticDDsixteen$\times$ and \resStaticWin$\times$. The training layers run a fixed number of iterations and gain nothing from release: label assignment gains \resOTApgBatching$\times$ from batching and \resOTApgPadded$\times$ more from the short-by-long kernel, mHC from packing, and patches from valid-size instances.

\subsection{Operating range}
\label{sec:overheads}

On synthetic point clouds with $d\le4$, $N=M\le512$, and $B\in\{1,16,1024\}$ (\cref{fig:regime} shows $d=4$), \drain{} is faster than OTT-JAX for every $N\le128$ and for every $N$ at $B=1024$. With $B\le16$, OTT-JAX is faster at $N=512$ and in half of the $N=256$ cells (\resPGLoss), and POT only at $N=512$ with $B=1$ and in one $B=16$ cell. All pass the gate.

\section{Related work}
\label{sec:related}

\paragraph{Fast Sinkhorn solvers.}
Entropic OT and its log-domain stabilization are standard \citep{cuturi2013sinkhorn,peyre2019computational,schmitzer2019stabilized}. Online reductions evaluate the cost on the fly \citep{feydy2019sinkhorn,charlier2021keops}, and FlashSinkhorn fuses the streamed log-sum-exp in an IO-aware kernel \citep{ye2026flashsinkhorn}. We build on both. FlashSinkhorn's batched interface loops over problems in Python, so we compare with it on single large problems (\cref{app:more}). POT and OTT-JAX batch and vectorize solvers \citep{flamary2021pot,cuturi2022ott}, and cuRegOT targets single large problems \citep{qiu2026curegot}. \citet{medina2025flow} solve the small cell problems of domain decomposition with a batched GPU solver that uses separable reductions \citep{logsinkhorngpu} and stops a batch on an aggregate error. We adopt the reduction and stop per instance. Training layers motivate shape-specific kernels: mHC fuses its $n=4$ projection into one kernel and recomputes it on chip in its backward \citep{xie2025mhc}. \citet{wang2026birkhoff} solve each $4\times4$ problem in half a warp's registers, and a finished problem waits for its block.

\paragraph{Batched execution of problems with different lengths.}
Finished items leave a batch in iteration-level scheduling of generation \citep{yu2022orca}, automatic batching \citep{radul2020autobatching}, batched ODE solvers \citep{lienen2022torchode}, interior-point solvers \citep{viljoen2026jaxipm}, and batched linear solvers \citep{nguyen2023ginkgosycl}. To our knowledge, no prior batched OT solver removes converged instances from later launches so that the grid narrows. Implicit gradients \citep{eisenberger2022implicit,pineda2022theseus} target the optimum, not the finite map.

\section{Limitations and conclusion}
\label{sec:conclusion}

Making the instance the unit of GPU execution serves inference and training alike: \drain{} is the fastest compared method on every tracking and domain-decomposition workload, up to \resMainOTTMax$\times$ and \resMainPOTMax$\times$ faster than OTT-JAX and POT, and cuts the Sinkhorn part of an OTA step, which YOLOX dropped, from \resOTApgSinkShareOfficial\% to \resOTApgSinkShareDrain\%.

Label-assignment inputs come from a converged detector, and timings cover single training steps. The mHC inputs come from a small model, and patch inputs from 32 3DMatch pairs. $4\times4$ projection kernels remain faster, mostly in the backward pass. Full $65\times65$ patches need $128\times128$ tiles, and packing covers fixed-iteration solves only. \drain{} loses to OTT-JAX on up to 16 clouds of 512 points and in half of the 256-point cells (\cref{sec:overheads}), and to OTT-JAX or POT on \resLargeSlower{} of 9 single large problems, by \resLargeHighDim$\times$ at $d=3072$, where it recomputes distances per half-step (\cref{app:more}).

\section*{Reproducibility statement}

\Cref{app:workloads} describes the workloads. The source code of \drain{} with its tests, and the benchmark (timing scripts, baseline configurations, and the scripts that capture the training-layer inputs), is available at \url{https://github.com/cis-aconiticacid/DrainSinkhorn}.

\section*{AI use statement}

The research idea, the problem setting, and the overall design of the solver and of this paper were provided by the author. Large language models (LLMs) were used in four roles. First, LLM coding agents wrote the GPU kernels of \drain{} and the benchmark code, including the baseline wrappers and the timing and comparison code. The author reviewed all of this code. Second, LLM coding agents wrote the remaining experiment code, including run scheduling and analysis scripts. LLMs reviewed it, and the author read the code behind the key steps of diagrams in which an LLM mapped each code block. Third, the author drafted the outline of the paper, an LLM wrote the text from that outline and from the experimental records, and the author revised the result into the final version. Fourth, LLM-based search tools were used to find related work: they listed candidate papers and suggested that some could be built into the solver. The author discarded some, tested the rest, and adopted the separable reduction of \cref{sec:separable} \citep{medina2025flow,logsinkhorngpu}. The author takes full responsibility for the content of the paper.

\bibliography{references}

\begin{thebibliography}{41}
\providecommand{\natexlab}[1]{#1}
\providecommand{\url}[1]{\texttt{#1}}
\expandafter\ifx\csname urlstyle\endcsname\relax
  \providecommand{\doi}[1]{doi: #1}\else
  \providecommand{\doi}{doi: \begingroup \urlstyle{rm}\Url}\fi

\bibitem[Ansel et~al.(2024)Ansel, Yang, He, Gimelshein, Jain, Voznesensky, Bao,
  Bell, Berard, Burovski, Chauhan, Chourdia, Constable, Desmaison, DeVito,
  Ellison, Feng, Gong, Gschwind, Hirsh, Huang, Kalambarkar, Kirsch, Lazos,
  Lezcano, Liang, Liang, Lu, Luk, Maher, Pan, Puhrsch, Reso, Saroufim,
  Siraichi, Suk, Zhang, Suo, Tillet, Zhao, Wang, Zhou, Zou, Wang, Mathews, Wen,
  Chanan, Wu, and Chintala]{ansel2024pytorch2}
Jason Ansel, Edward Yang, Horace He, Natalia Gimelshein, Animesh Jain, Michael
  Voznesensky, Bin Bao, Peter Bell, David Berard, Evgeni Burovski, Geeta
  Chauhan, Anjali Chourdia, Will Constable, Alban Desmaison, Zachary DeVito,
  Elias Ellison, Will Feng, Jiong Gong, Michael Gschwind, Brian Hirsh, Sherlock
  Huang, Kshiteej Kalambarkar, Laurent Kirsch, Michael Lazos, Mario Lezcano,
  Yanbo Liang, Jason Liang, Yinghai Lu, C.~K. Luk, Bert Maher, Yunjie Pan,
  Christian Puhrsch, Matthias Reso, Mark Saroufim, Marcos~Yukio Siraichi, Helen
  Suk, Shunting Zhang, Michael Suo, Phil Tillet, Xu~Zhao, Eikan Wang, Keren
  Zhou, Richard Zou, Xiaodong Wang, Ajit Mathews, William Wen, Gregory Chanan,
  Peng Wu, and Soumith Chintala.
\newblock {PyTorch} 2: Faster machine learning through dynamic {Python}
  bytecode transformation and graph compilation.
\newblock In \emph{ACM International Conference on Architectural Support for
  Programming Languages and Operating Systems (ASPLOS)}, 2024.

\bibitem[Bonafini \& Schmitzer(2021)Bonafini and Schmitzer]{bonafini2021domain}
Mauro Bonafini and Bernhard Schmitzer.
\newblock Domain decomposition for entropy regularized optimal transport.
\newblock \emph{Numerische Mathematik}, 149:\penalty0 819--870, 2021.

\bibitem[Charlier et~al.(2021)Charlier, Feydy, Glaun{\`e}s, Collin, and
  Durif]{charlier2021keops}
Benjamin Charlier, Jean Feydy, Joan~Alexis Glaun{\`e}s, Fran{\c{c}}ois-David
  Collin, and Ghislain Durif.
\newblock Kernel operations on the {GPU}, with autodiff, without memory
  overflows.
\newblock \emph{Journal of Machine Learning Research}, 22\penalty0
  (74):\penalty0 1--6, 2021.

\bibitem[Chen et~al.(2025)Chen, Shen, Deng, Zhao, Liang, and
  Wong]{chen2025pearl}
Liang Chen, Li~Shen, Yang Deng, Xiaoyan Zhao, Bin Liang, and Kam-Fai Wong.
\newblock {PEARL}: Towards permutation-resilient {LLMs}.
\newblock In \emph{International Conference on Learning Representations
  (ICLR)}, 2025.

\bibitem[Cuturi(2013)]{cuturi2013sinkhorn}
Marco Cuturi.
\newblock Sinkhorn distances: Lightspeed computation of optimal transportation
  distances.
\newblock In \emph{Advances in Neural Information Processing Systems},
  volume~26, pp.\  2292--2300, 2013.
\newblock URL \url{https://arxiv.org/abs/1306.0895}.

\bibitem[Cuturi et~al.(2022)Cuturi, Meng-Papaxanthos, Tian, Bunne, Davis, and
  Teboul]{cuturi2022ott}
Marco Cuturi, Laetitia Meng-Papaxanthos, Yingtao Tian, Charlotte Bunne, Geoff
  Davis, and Olivier Teboul.
\newblock Optimal transport tools ({OTT}): A {JAX} toolbox for all things
  wasserstein, 2022.
\newblock URL \url{https://arxiv.org/abs/2201.12324}.

\bibitem[{DeepSeek-AI}(2026)]{deepseek2026v4}
{DeepSeek-AI}.
\newblock {DeepSeek-V4}: Towards highly efficient million-token context
  intelligence.
\newblock \emph{arXiv preprint arXiv:2606.19348}, 2026.

\bibitem[Eisenberger et~al.(2022)Eisenberger, Toker, Leal-Taix{\'e}, Bernard,
  and Cremers]{eisenberger2022implicit}
Marvin Eisenberger, Aysim Toker, Laura Leal-Taix{\'e}, Florian Bernard, and
  Daniel Cremers.
\newblock A unified framework for implicit {Sinkhorn} differentiation.
\newblock In \emph{Proceedings of the IEEE/CVF Conference on Computer Vision
  and Pattern Recognition}, pp.\  509--518, 2022.

\bibitem[Feydy et~al.(2019)Feydy, S{\'e}journ{\'e}, Vialard, ichi Amari,
  Trouv{\'e}, and Peyr{\'e}]{feydy2019sinkhorn}
Jean Feydy, Thibault S{\'e}journ{\'e}, Fran{\c{c}}ois-Xavier Vialard, Shun ichi
  Amari, Alain Trouv{\'e}, and Gabriel Peyr{\'e}.
\newblock Interpolating between optimal transport and {MMD} using sinkhorn
  divergences.
\newblock In \emph{Proceedings of the 22nd International Conference on
  Artificial Intelligence and Statistics}, volume~89 of \emph{Proceedings of
  Machine Learning Research}, pp.\  2681--2690, 2019.
\newblock URL \url{https://proceedings.mlr.press/v89/feydy19a.html}.

\bibitem[Flamary et~al.(2021)Flamary, Courty, Gramfort, Alaya, Boisbunon,
  Chambon, Chapel, Corenflos, Fatras, Fournier, Gautheron, Gayraud, Janati,
  Rakotomamonjy, Redko, Rolet, Schutz, Seguy, Sutherland, Tavenard, Tong, and
  Vayer]{flamary2021pot}
R{\'e}mi Flamary, Nicolas Courty, Alexandre Gramfort, Mokhtar~Z. Alaya,
  Aur{\'e}lie Boisbunon, Stanislas Chambon, Laetitia Chapel, Adrien Corenflos,
  Kilian Fatras, Nemo Fournier, L{\'e}o Gautheron, Nathalie T.~H. Gayraud,
  Hicham Janati, Alain Rakotomamonjy, Ievgen Redko, Antoine Rolet, Antony
  Schutz, Vivien Seguy, Danica~J. Sutherland, Romain Tavenard, Alexander Tong,
  and Titouan Vayer.
\newblock {POT}: {Python} optimal transport.
\newblock \emph{Journal of Machine Learning Research}, 22\penalty0
  (78):\penalty0 1--8, 2021.
\newblock URL \url{https://www.jmlr.org/papers/v22/20-451.html}.

\bibitem[Ge et~al.(2021{\natexlab{a}})Ge, Liu, Li, Yoshie, and Sun]{ge2021ota}
Zheng Ge, Songtao Liu, Zeming Li, Osamu Yoshie, and Jian Sun.
\newblock {OTA}: Optimal transport assignment for object detection.
\newblock In \emph{IEEE/CVF Conference on Computer Vision and Pattern
  Recognition (CVPR)}, 2021{\natexlab{a}}.

\bibitem[Ge et~al.(2021{\natexlab{b}})Ge, Liu, Wang, Li, and Sun]{ge2021yolox}
Zheng Ge, Songtao Liu, Feng Wang, Zeming Li, and Jian Sun.
\newblock {YOLOX}: Exceeding {YOLO} series in 2021.
\newblock \emph{arXiv preprint arXiv:2107.08430}, 2021{\natexlab{b}}.

\bibitem[Graham(1969)]{graham1969bounds}
Ronald~L. Graham.
\newblock Bounds on multiprocessing timing anomalies.
\newblock \emph{SIAM Journal on Applied Mathematics}, 17\penalty0 (2):\penalty0
  416--429, 1969.

\bibitem[Lienen \& G{\"u}nnemann(2022)Lienen and
  G{\"u}nnemann]{lienen2022torchode}
Marten Lienen and Stephan G{\"u}nnemann.
\newblock torchode: A parallel {ODE} solver for {PyTorch}, 2022.
\newblock URL \url{https://arxiv.org/abs/2210.12375}.

\bibitem[Lindenberger et~al.(2023)Lindenberger, Sarlin, and
  Pollefeys]{lindenberger2023lightglue}
Philipp Lindenberger, Paul-Edouard Sarlin, and Marc Pollefeys.
\newblock {LightGlue}: Local feature matching at light speed.
\newblock In \emph{IEEE/CVF International Conference on Computer Vision
  (ICCV)}, 2023.

\bibitem[Liu et~al.(2026)Liu, Zhang, and Li]{liu2026spectral}
Zhaoyi Liu, Haichuan Zhang, and Ang Li.
\newblock Spectral-sphere-constrained hyper-connections.
\newblock In \emph{Advances in Neural Information Processing Systems
  (NeurIPS)}, 2026.

\bibitem[{LottoLottoLotto}(2026)]{tritonsinkhorn}
{LottoLottoLotto}.
\newblock triton-sinkhorn.
\newblock \url{https://github.com/LottoLottoLotto/triton-sinkhorn}, 2026.
\newblock Commit c5bd2f6d, accessed 2026-09-25.

\bibitem[Medina \& Schmitzer(2025)Medina and Schmitzer]{medina2025flow}
Ismael Medina and Bernhard Schmitzer.
\newblock Flow updates for domain decomposition of entropic optimal transport.
\newblock \emph{ESAIM: Mathematical Modelling and Numerical Analysis},
  59\penalty0 (3), 2025.
\newblock \doi{10.1051/m2an/2025022}.

\bibitem[M{\'e}ndez et~al.(2026)M{\'e}ndez, Krzakala, Melo, Laclau, Flamary,
  and d'Alch{\'e} Buc]{mendez2026graph}
Federico M{\'e}ndez, Paul Krzakala, Gabriel Melo, Charlotte Laclau, R{\'e}mi
  Flamary, and Florence d'Alch{\'e} Buc.
\newblock Graph matching relaxations and amortization for supervised graph
  prediction.
\newblock \emph{arXiv preprint arXiv:2609.15437}, 2026.

\bibitem[Milan et~al.(2016)Milan, Leal-Taix{\'e}, Reid, Roth, and
  Schindler]{milan2016mot16}
Anton Milan, Laura Leal-Taix{\'e}, Ian Reid, Stefan Roth, and Konrad Schindler.
\newblock {MOT16}: A benchmark for multi-object tracking, 2016.
\newblock URL \url{https://arxiv.org/abs/1603.00831}.

\bibitem[Na et~al.(2026)Na, Shin, Kim, Kang, and Moon]{na2026wasd}
Byeonghu Na, Donghyeok Shin, Yeongmin Kim, Mina Kang, and Il-Chul Moon.
\newblock {WASD}: Wasserstein-based knowledge distillation for large language
  models.
\newblock In \emph{Advances in Neural Information Processing Systems
  (NeurIPS)}, 2026.

\bibitem[Nguyen et~al.(2023)Nguyen, Nayak, and Anzt]{nguyen2023ginkgosycl}
Phuong Nguyen, Pratik Nayak, and Hartwig Anzt.
\newblock Porting batched iterative solvers onto {Intel} {GPUs} with {SYCL},
  2023.
\newblock URL \url{https://arxiv.org/abs/2308.08417}.

\bibitem[{OT Group G{\"o}ttingen}(2025)]{logsinkhorngpu}
{OT Group G{\"o}ttingen}.
\newblock {LogSinkhornGPU}: {GPU} log-domain {Sinkhorn} solvers for domain
  decomposition.
\newblock \url{https://github.com/OTGroupGoe/LogSinkhornGPU}, 2025.
\newblock Commit 58ac43cc.

\bibitem[Peyr{\'e} \& Cuturi(2019)Peyr{\'e} and Cuturi]{peyre2019computational}
Gabriel Peyr{\'e} and Marco Cuturi.
\newblock Computational optimal transport.
\newblock \emph{Foundations and Trends in Machine Learning}, 11\penalty0
  (5--6):\penalty0 355--607, 2019.

\bibitem[Pineda et~al.(2022)Pineda, Fan, Monge, Venkataraman, Sodhi, Chen,
  Ortiz, DeTone, Wang, Anderson, Dong, Amos, and Mukadam]{pineda2022theseus}
Luis Pineda, Taosha Fan, Maurizio Monge, Shobha Venkataraman, Paloma Sodhi,
  Ricky T.~Q. Chen, Joseph Ortiz, Daniel DeTone, Austin Wang, Stuart Anderson,
  Jing Dong, Brandon Amos, and Mustafa Mukadam.
\newblock Theseus: A library for differentiable nonlinear optimization.
\newblock In \emph{Advances in Neural Information Processing Systems},
  volume~35, 2022.

\bibitem[Qin et~al.(2022)Qin, Yu, Wang, Guo, Peng, and
  Xu]{qin2022geotransformer}
Zheng Qin, Hao Yu, Changjian Wang, Yulan Guo, Yuxing Peng, and Kai Xu.
\newblock Geometric transformer for fast and robust point cloud registration.
\newblock In \emph{IEEE/CVF Conference on Computer Vision and Pattern
  Recognition}, 2022.

\bibitem[Qiu(2026)]{qiu2026curegot}
Yixuan Qiu.
\newblock {cuRegOT}: A {GPU}-accelerated solver for entropic-regularized
  optimal transport, 2026.
\newblock URL \url{https://arxiv.org/abs/2605.08793}.

\bibitem[Radul et~al.(2020)Radul, Patton, Maclaurin, Hoffman, and
  Saurous]{radul2020autobatching}
Alexey Radul, Brian Patton, Dougal Maclaurin, Matthew~D. Hoffman, and Rif~A.
  Saurous.
\newblock Automatically batching control-intensive programs for modern
  accelerators.
\newblock In \emph{Proceedings of Machine Learning and Systems}, volume~2,
  2020.
\newblock URL
  \url{https://proceedings.mlsys.org/paper_files/paper/2020/hash/39945d578f616735572174bf5e8f155d-Abstract.html}.

\bibitem[Schinagl et~al.(2026)Schinagl, Fruhwirth-Reisinger, Prutsch, Schulter,
  and Possegger]{schinagl2026rank}
David Schinagl, Christian Fruhwirth-Reisinger, Alexander Prutsch, Samuel
  Schulter, and Horst Possegger.
\newblock Learn to rank: Visual attribution by learning importance ranking.
\newblock In \emph{European Conference on Computer Vision (ECCV)}, 2026.

\bibitem[Schmitzer(2019)]{schmitzer2019stabilized}
Bernhard Schmitzer.
\newblock Stabilized sparse scaling algorithms for entropy regularized
  transport problems.
\newblock \emph{SIAM Journal on Scientific Computing}, 41\penalty0
  (3):\penalty0 A1443--A1481, 2019.

\bibitem[Tian et~al.(2019)Tian, Shen, Chen, and He]{tian2019fcos}
Zhi Tian, Chunhua Shen, Hao Chen, and Tong He.
\newblock {FCOS}: Fully convolutional one-stage object detection.
\newblock In \emph{IEEE/CVF International Conference on Computer Vision
  (ICCV)}, 2019.

\bibitem[Viljoen et~al.(2026)Viljoen, Haffner, Tomizuka, and
  Mehr]{viljoen2026jaxipm}
John Viljoen, Johanna Haffner, Masayoshi Tomizuka, and Negar Mehr.
\newblock Scaling nonlinear optimization: Many problems one {GPU}, 2026.
\newblock URL \url{https://arxiv.org/abs/2606.26341}.

\bibitem[Wang et~al.(2020)Wang, Liu, Samaras, and Hoai]{wang2020dmcount}
Boyu Wang, Huidong Liu, Dimitris Samaras, and Minh Hoai.
\newblock Distribution matching for crowd counting.
\newblock In \emph{Advances in Neural Information Processing Systems},
  volume~33, 2020.

\bibitem[Wang \& Qiu(2026)Wang and Qiu]{wang2026birkhoff}
Chenrui Wang and Yixuan Qiu.
\newblock Accelerating {Birkhoff} projection for manifold-constrained
  hyper-connections, 2026.
\newblock URL \url{https://arxiv.org/abs/2606.07574}.

\bibitem[Wang et~al.(2025)Wang, Cheng, Shi, Tang, Mo, Xie, Ma, Xia, Xue, Yang,
  and Yang]{wang2025tilelang}
Lei Wang, Yu~Cheng, Yining Shi, Zhengju Tang, Zhiwen Mo, Wenhao Xie, Lingxiao
  Ma, Yuqing Xia, Jilong Xue, Fan Yang, and Zhi Yang.
\newblock {TileLang}: A composable tiled programming model for {AI} systems.
\newblock \emph{arXiv preprint arXiv:2504.17577}, 2025.

\bibitem[Wang et~al.(2026)Wang, Zhou, Wang, Tan, Zhao, Li{\`o}, and
  Hong]{wang2026plasma}
Zhiyu Wang, Bingxin Zhou, Jing Wang, Yang Tan, Weishu Zhao, Pietro Li{\`o}, and
  Liang Hong.
\newblock Fast and interpretable protein substructure alignment via optimal
  transport.
\newblock In \emph{International Conference on Learning Representations
  (ICLR)}, 2026.

\bibitem[Xie et~al.(2025)Xie, Wei, Cao, Zhao, Deng, Li, Dai, Gao, Chang, Yu,
  Zhao, Zhou, Xu, Zhang, Zeng, Hu, Wang, Yuan, Wang, and Liang]{xie2025mhc}
Zhenda Xie, Yixuan Wei, Huanqi Cao, Chenggang Zhao, Chengqi Deng, Jiashi Li,
  Damai Dai, Huazuo Gao, Jiang Chang, Kuai Yu, Liang Zhao, Shangyan Zhou, Zhean
  Xu, Zhengyan Zhang, Wangding Zeng, Shengding Hu, Yuqing Wang, Jingyang Yuan,
  Lean Wang, and Wenfeng Liang.
\newblock {mHC}: Manifold-constrained hyper-connections.
\newblock \emph{arXiv preprint arXiv:2512.24880}, 2025.

\bibitem[Ye et~al.(2026)Ye, Li, Yu, Chang, Chu, and
  Wertheimer]{ye2026flashsinkhorn}
Felix X.-F. Ye, Xingjie Li, An~Yu, Ming-Ching Chang, Linsong Chu, and Davis
  Wertheimer.
\newblock {FlashSinkhorn}: {IO}-aware entropic optimal transport on {GPU},
  2026.
\newblock URL \url{https://arxiv.org/abs/2602.03067}.

\bibitem[Yu et~al.(2022)Yu, Jeong, Kim, Kim, and Chun]{yu2022orca}
Gyeong-In Yu, Joo~Seong Jeong, Geon-Woo Kim, Soojeong Kim, and Byung-Gon Chun.
\newblock {Orca}: A distributed serving system for transformer-based generative
  models.
\newblock In \emph{16th USENIX Symposium on Operating Systems Design and
  Implementation}, pp.\  521--538, 2022.
\newblock URL \url{https://www.usenix.org/conference/osdi22/presentation/yu}.

\bibitem[Yu et~al.(2021)Yu, Li, Saleh, Busam, and Ilic]{yu2021cofinet}
Hao Yu, Fu~Li, Mahdi Saleh, Benjamin Busam, and Slobodan Ilic.
\newblock {CoFiNet}: Reliable coarse-to-fine correspondences for robust point
  cloud registration.
\newblock In \emph{Advances in Neural Information Processing Systems}, 2021.

\bibitem[Zhang et~al.(2024)Zhang, Xue, Fan, Xu, Li, Pei, and
  Liu]{zhang2024ltgnn}
Jiahao Zhang, Rui Xue, Wenqi Fan, Xin Xu, Qing Li, Jian Pei, and Xiaorui Liu.
\newblock Linear-time graph neural networks for scalable recommendations.
\newblock In \emph{Proceedings of the ACM Web Conference (WWW)}, 2024.

\end{thebibliography}
\bibliographystyle{iclr2027_conference}

\appendix
\section{Proofs}
\label{app:proofs}

Throughout, arithmetic is exact.

\paragraph{Proof of \eqref{eq:work}.}
$\sum_t|\Aset_t|=\sum_t\sum_b\mathbf 1[D_b\ge t]=\sum_b\sum_{t=1}^{\Dmax}\mathbf 1[t\le D_b]=\sum_b D_b$. \hfill$\square$

\paragraph{Proof of \cref{prop:makespan}.}
Write $p_b=\theta D_b+\sigma$ and $W=\sum_b p_b$. At most $S$ programs run at once, so $T_{\mathrm{active}}\ge W/S$, and no program finishes before $p_{\max}=\theta\Dmax+\sigma$ has elapsed. For the upper bound, let program $j$ finish last and start at $s_j$. Because the scheduler is work-conserving and $j$ was pending until $s_j$, all $S$ slots were busy throughout $[0,s_j)$, so $Ss_j\le W-p_j$. Then $T_{\mathrm{active}}=s_j+p_j\le W/S+(1-1/S)p_j\le W/S+(1-1/S)p_{\max}$ \citep{graham1969bounds}. In the static control all $B$ programs have length $p_{\max}$. A work-conserving schedule of equal jobs runs them in $\lceil B/S\rceil$ waves, so $T_{\mathrm{static}}=\lceil B/S\rceil p_{\max}$. If $B\le S$, $T_{\mathrm{static}}=p_{\max}\le T_{\mathrm{active}}$. If the depths are drawn from a fixed distribution with mean $\bar D$ and $B\to\infty$, both $T_{\mathrm{static}}$ and $T_{\mathrm{active}}$ grow like $B/S$ times $p_{\max}$ and $\theta\bar D+\sigma$ respectively. \hfill$\square$

The assumptions exclude contention between concurrent programs. Each resident program loads its tile once and then works on chip, so memory traffic is limited to the initial loads and final stores, but co-resident programs on one SM still share its issue bandwidth.

\section{Implementation details}
\label{app:impl}

\begin{algorithm}[h]
\caption{Per-instance release (both paths return $f,g,r,D$, status in original order)}
\label{alg:drain}
\small
\begin{algorithmic}[1]
\STATE \textbf{Resident path} \hfill\textit{one program per instance $b$, $B$ programs in one launch}
\STATE load $C_b$, $a_b$, $b_b$ on chip; $f,g\leftarrow 0$; $t\leftarrow 0$
\REPEAT
  \STATE apply \eqref{eq:updates} $\min(q,K_{\max}-t)$ times; $t\leftarrow t+q$
  \STATE $P\leftarrow\exp((f\oplus g-C_b)/\eps_b)$ on chip; $r\leftarrow\max(\|P\ones-a_b\|_1,\|P^\top\ones-b_b\|_1)$
\UNTIL{$r\le\tau$ \textbf{or} $t\ge K_{\max}$}
\STATE write $f_b,g_b,r_b,D_b\leftarrow t$ and exit \hfill\textit{the scheduler reuses the slot}
\STATE \textbf{Indexed path} \hfill\textit{host loop over launches}
\STATE $\Aset\leftarrow\{1,\dots,B\}$; $f,g\leftarrow 0$; \textit{reuse}$\leftarrow$false
\FOR{$t=1,\dots,K_{\max}$ \textbf{while} $\Aset\neq\emptyset$}
  \STATE $f_{\Aset}\leftarrow f^{\mathrm{bal}}_{\Aset}$ \textbf{if} \textit{reuse} \textbf{else} row update \hfill\textit{survivors reuse the screen's update}
  \STATE $g_{\Aset}\leftarrow$ column update; \textit{reuse}$\leftarrow$false
  \IF{$t\equiv 0 \pmod q$}
    \STATE $f^{\mathrm{bal}}_{\Aset},s_{\Aset}\leftarrow$ row update from current $g$, with screen epilogue \eqref{eq:screen}
    \STATE $V\leftarrow\{b\in\Aset: s_b\le\tau\}$; $r_V\leftarrow$ both residuals of current $(f_V,g_V)$
    \STATE store $(f_b,g_b,r_b,t)$ for accepted $b\in V$; $\Aset\leftarrow\Aset\setminus\{b\in V:r_b\le\tau\}$; \textit{reuse}$\leftarrow$true
  \ENDIF
\ENDFOR
\STATE verify remaining $\Aset$ once and mark them capped
\end{algorithmic}
\end{algorithm}

\paragraph{Routing.}
A tolerance-mode FP32 CUDA solve without gradients uses the resident path if both supports have between 1 and 64 points. Squared-Euclidean supports with $d\le4$ and up to 512 points use the low-dimensional resident path when its tiling search (\cref{sec:resident}) selects it. A fixed-iteration dense FP32 CUDA solve with a per-instance cost uses the packed resident map if both sides have at most 128 points, or if the short side has at most 64 and the padded tile at most $2^{17}$ elements. It supports gradients with respect to all inputs. Other calls use the indexed path. The separable option is selected explicitly by passing the grid shapes of both supports, each axis with at most 16 points. On the A100, the compiled resident kernel admits at most 1296, 972, and 324 concurrent programs for 16-, 32-, and 64-point instances. These are the values of $S$ used in \cref{fig:a1}.

\paragraph{Cost preparation.}
Coordinates are centered per instance and squared norms are cached once per call. The squared distance is computed as $\|x\|^2+\|y\|^2-2\langle x,y\rangle$ in FP32 without TF32. For the resident path the tile is materialized once. A support shared by all instances is stored once and read with batch stride zero. The separable option stores only the two axis cost tiles of each instance.

\paragraph{Indexed kernels.}
Each half-step launches a grid over (active instance, row tile). A program streams column tiles, computes the distance tile by a tiled dot product, and maintains an online log-sum-exp. The dot-product and online log-sum-exp helpers are adapted from FlashSinkhorn \citep{ye2026flashsinkhorn} under its MIT license, and the consumer kernel, which applies the plan rebuilt from the potentials without storing it, follows the organization of FlashSinkhorn's plan-apply kernel. Tile shapes are chosen by autotuning keyed on instance size and active width. The screen is a flag of the row-update kernel that additionally accumulates $a_i|\expmone((f_i-f^{\mathrm{bal}}_i)/\eps)|$. A separate kernel computes both marginals for verification.

\paragraph{Fixed-iteration kernels.}
The packed forward kernel of \cref{sec:fixed} holds 32 instances per program for padded tiles of at most 64 elements and $\max(1,\lfloor2048/(RC)\rfloor)$ for larger tiles of $R\times C$ elements, with four warps. Input value checks (finite cost, nonnegative marginals of equal mass, positive $\eps$) run as one packed kernel whose error flags are read once, and \texttt{validate=False} skips them. Tests compare the forward outputs and the vector--Jacobian products with respect to $C$, $a$, $b$, and $\eps$ of the packed path with FP64 automatic differentiation of the same $K$-step map. The row and column residuals of the returned iterate are computed by one more packed kernel.

\section{Inference workloads}
\label{app:more}

\subsection{Workloads}
\label{app:workloads}
\Cref{tab:workloads} lists the seven inference workloads.
\begin{table}[h]
\centering
\caption{\textbf{Workloads.} All instances are balanced squared-Euclidean EOT problems with $\tau=10^{-4}$ and a check every $q=5$ iterations. $B$ is the number of instances per batch. MOT17 buckets pad each detection set to the bucket size with points of mass $10^{-8}$. Domain-decomposition tiles share one grid support.}
\label{tab:workloads}
\small
\setlength{\tabcolsep}{4.5pt}
\begin{tabular}{llrrrrrll}
\toprule
Workload & Source & $B$ & $N{=}M$ & $d$ & $\eps$ & $K_{\max}$ & Supports & Path \\
\midrule
\texttt{mot-b16}   & MOT17 frame pairs & 3771 & 16   & 4 & 0.05 & 500  & independent & resident \\
\texttt{mot-b32}   & MOT17 frame pairs & 1469 & 32   & 4 & 0.05 & 500  & independent & resident \\
\texttt{mot-b64}   & MOT17 frame pairs & 69   & 64   & 4 & 0.05 & 500  & independent & resident \\
\texttt{mot-win}   & MOT17 windows     & 64   & 1024 & 4 & 0.05 & 500  & independent & indexed \\
\texttt{dd-t4}     & DD tile pool      & 2304 & 16   & 2 & 0.01 & 2000 & shared grid & resident \\
\texttt{dd-t8}     & DD tile pool      & 576  & 64   & 2 & 0.01 & 500  & shared grid & resident \\
\texttt{dd-t16}    & DD tile pool      & 144  & 256  & 2 & 0.01 & 500  & shared grid & resident \\
\bottomrule
\end{tabular}
\end{table}

\subsection{Other GPUs}
\label{app:gpus}
We repeated the inference workloads on an NVIDIA A100-SXM4-80GB, an H100 80GB HBM3, and an L40S in one campaign, together with the training layers (\cref{app:gpus-train}): the same solver with the constants chosen on the A100, inputs rebuilt by the same scripts, Triton 3.7.1, PyTorch \resTorch, and 5 to 10 paired blocks per cell. Every output passes the gates of \cref{sec:setup} on every GPU. \drain{} is the fastest compared method on every workload and GPU, \gpuOTTA{}$\times$ faster than OTT-JAX on the A100, \gpuOTTH{}$\times$ on the H100, and \gpuOTTL{}$\times$ on the L40S, and \gpuPOTRange{}$\times$ faster than POT (\cref{tab:gpus}).

\begin{table}[h]
\centering
\caption{Inference workloads on three GPUs: OTT-JAX and POT time over \drain{}, paired medians per block.}
\label{tab:gpus}
\small
\begin{tabular}{lrrrrrr}
\toprule
& \multicolumn{2}{c}{A100} & \multicolumn{2}{c}{H100} & \multicolumn{2}{c}{L40S} \\
\cmidrule(lr){2-3}\cmidrule(lr){4-5}\cmidrule(lr){6-7}
Workload & OTT & POT & OTT & POT & OTT & POT \\
\midrule
\texttt{dd-t4} & 11.1 & 47.7 & 11.6 & 46.1 & 13.4 & 46.7 \\
\texttt{dd-t8} & 5.00 & 10.7 & 4.32 & 10.5 & 6.33 & 11.2 \\
\texttt{dd-t16} & 2.64 & 5.22 & 2.21 & 5.26 & 1.70 & 5.18 \\
\texttt{mot-b16} & 10.6 & 19.2 & 5.56 & 15.5 & 8.06 & 17.9 \\
\texttt{mot-b32} & 8.99 & 22.1 & 6.59 & 20.0 & 10.0 & 23.1 \\
\texttt{mot-b64} & 3.12 & 11.6 & 2.36 & 9.43 & 2.93 & 11.0 \\
\texttt{mot-win} & 3.08 & 5.89 & 2.62 & 5.61 & 10.1 & 21.3 \\
 
\bottomrule
\end{tabular}
\end{table}

\paragraph{Effect of the GPU.}
A resident instance must fit in the registers of one program. The register file is the same on all three GPUs, 64K 32-bit registers per SM and at most 255 per thread, so the size limits of the resident paths do not change. The number of programs that run at once does: an SM holds at most 32 blocks on the A100 and H100 and 24 on the L40S, and the three GPUs have 108, 132, and 142 SMs. This sets how many waves a batch needs, and with it the gain from release (\cref{prop:makespan}). \texttt{mot-b64} fits in one wave and gains 1.00--1.03$\times$ from release on every GPU (\cref{tab:source}). The packing constants of the fixed-iteration path stay at their A100 values, while the tiling search of the low-dimensional resident path and the autotuning of the indexed path run on each device.

\paragraph{First call.}
\Cref{tab:cold} times the first call of each workload in a fresh process, with an empty and with a warm compilation cache. On \texttt{dd-t16} and \texttt{mot-win}, the first call searches resident tilings or autotunes tile shapes. Each process and shape pays this once, so a training loop spreads it over many calls, while a single call pays it in full.

\begin{table}[h]
\centering
\caption{First call in a fresh process (s, range over workloads and GPUs). All outputs pass the gate.}
\label{tab:cold}
\small
\begin{tabular}{lrrrr}
\toprule
& \multicolumn{2}{c}{\drain} & \multicolumn{2}{c}{OTT-JAX} \\
\cmidrule(lr){2-3}\cmidrule(lr){4-5}
Workloads & empty cache & warm cache & empty cache & warm cache \\
\midrule
\texttt{dd-t16}, \texttt{mot-win} & 38.9--49.4 & 8.96--11.1 & 1.64--2.24 & 0.60--1.09 \\
other five workloads & 2.75--4.52 & 0.58--0.88 & 1.19--2.01 & 0.46--0.70 \\
 
\bottomrule
\end{tabular}
\end{table}

\subsection{Launches and single large problems}
Under Nsight Systems, one \texttt{mot-b16} call issues \resProfLaunchDrain{} kernel launches for \drain{} against \resProfLaunchOTT{} for OTT-JAX and \resProfLaunchPOT{} for POT, with 2.1~ms of kernel time against 39.1 and 53.8~ms.

For one problem ($B=1$) with $N\in\{1024,4096,16384\}$ and $d\in\{3,64,3072\}$, \drain{} is slower than the fastest of POT, OTT-JAX, and FlashSinkhorn in \resLargeSlower{} of 9 cases, by \resLargeHighDim$\times$ at $d=3072$, and faster only for $d=3$ at $N=4096$ and $16384$ (\resLargeFaster$\times$). FlashSinkhorn runs with its own early stop, a threshold of $\tau\eps$ on the change of the potentials checked every 5 iterations. It is never the fastest, and its stop returns residuals of $7\cdot10^{-7}$ to $7\cdot10^{-6}$. These runs use one fresh process per method and 20 repeats (5 at $N=16384$), without paired intervals.

\subsection{Regular grids}
\begin{table}[t]
\centering
\caption{\textbf{Regular-grid tiles:} \drain{} with the separable option against the native CUDA solver of LogSinkhornGPU \citep{logsinkhorngpu} (ms, \resOuterBlocks{} paired blocks). Ratio: native time over \drain{}, with 95\% CI. \emph{Capped}: instances at the cap.}
\label{tab:structured}
\small
\setlength{\tabcolsep}{5pt}
\begin{tabular}{lrrlrr}
\toprule
& & & & \multicolumn{2}{c}{Capped} \\
\cmidrule(lr){5-6}
Tiles & \drain & Native & Ratio [95\% CI] & \drain & Native \\
\midrule
\texttt{dd-t4} & 9.70 & 625 & 64.5\,[63.8, 65.2] & 343 & 2304 \\
\texttt{dd-t8} & 3.54 & 62.5 & 17.6\,[17.3, 17.8] & 0 & 0 \\
\texttt{dd-t16} & 5.29 & 58.7 & 11.1\,[11.0, 11.2] & 0 & 0 \\
 
\bottomrule
\end{tabular}
\end{table}

The separable option is \resSepRange$\times$ faster than the native LogSinkhornGPU solver on \texttt{dd-t8} and \texttt{dd-t16}, where both converge within the cap (\cref{tab:structured}). Both use the same reduction, so the difference lies in execution and stopping: one program per tile with per-instance release, against batch-wide launches until an aggregate error passes. On \texttt{dd-t4} the aggregate stop runs all 2304 tiles to the cap while \drain{} releases all but 343, so the \resSepDDfour$\times$ ratio includes a different amount of work.

\subsection{Iteration budget}
\begin{figure}[t]
\centering
\pgfplotsset{reg/.style={width=\linewidth, height=3.3cm, axis line style={black!60}, tick style={black!60},
  xticklabel style={font=\tiny}, yticklabel style={font=\tiny}, label style={font=\scriptsize},
  legend style={font=\tiny, draw=none, fill=white, fill opacity=0.8, text opacity=1}}}
\begin{subfigure}[t]{0.325\linewidth}\centering
\begin{tikzpicture}
\begin{axis}[reg, xmode=log, log basis x=2, xmin=400, xmax=5000, xtick={500,1000,2000,4000},
  xticklabels={500,1000,2000,4000}, ymin=0, ymax=75, xlabel={iteration cap $K_{\max}$},
  ylabel={POT / ours}, legend style={at={(0.5,1.02)}, anchor=south, legend columns=3}, unbounded coords=jump]
\addplot[mark=square*, mark size=1.3pt, draingreen] table[x=cap, y=ddt4, col sep=comma] {figures/data/plot_fig4a.csv};
\addlegendentry{dd-t4}
\addplot[mark=*, mark size=1.3pt, drainblue] table[x=cap, y=motb16, col sep=comma] {figures/data/plot_fig4a.csv};
\addlegendentry{mot-b16}
\addplot[mark=triangle*, mark size=1.5pt, drainorange] table[x=cap, y=motb32, col sep=comma] {figures/data/plot_fig4a.csv};
\addlegendentry{mot-b32}
\end{axis}
\end{tikzpicture}
\caption{Iteration budget.}\label{fig:cap}
\end{subfigure}\hfill
\begin{subfigure}[t]{0.325\linewidth}\centering
\begin{tikzpicture}
\begin{axis}[reg, ybar stacked, bar width=6pt, xmin=-0.6, xmax=6.6, ymin=0, ymax=100, scaled y ticks=false, xtick={0,...,6},
  xticklabels={dd-t4, dd-t8, dd-t16, mot-b16, mot-b32, mot-b64, mot-win},
  xticklabel style={rotate=60, anchor=east, inner sep=1pt}, ylabel={\% of wall time},
  legend style={at={(0.5,1.02)}, anchor=south, legend columns=2}]
\addplot[fill=drainblue!75, draw=none] table[x=x, y=solve_pct, col sep=comma] {figures/data/plot_fig4b.csv};
\addlegendentry{solve}
\addplot[fill=drainorange!75, draw=none] table[x=x, y=consumer_pct, col sep=comma] {figures/data/plot_fig4b.csv};
\addlegendentry{consumer}
\end{axis}
\end{tikzpicture}
\caption{Solver and consumer intervals.}\label{fig:breakdown}
\end{subfigure}\hfill
\begin{subfigure}[t]{0.325\linewidth}\centering
\begin{tikzpicture}
\begin{axis}[reg, xmode=log, log basis x=2, xmin=12, xmax=640, xtick={16,32,64,128,256,512},
  xticklabels={16,32,64,128,256,512}, ymode=log, ymin=0.3, ymax=32, ytick={0.5,1,2,4,8,16},
  yticklabels={0.5,1,2,4,8,16}, xlabel={support size $N=M$}, ylabel={OTT-JAX / ours},
  legend style={at={(0.5,1.02)}, anchor=south, legend columns=3}, legend image post style={xscale=0.6}]
\addplot[mark=o, mark size=1.3pt, draingray] table[x=n, y=b1, col sep=comma] {figures/data/plot_fig4c.csv};
\addlegendentry{$B{=}1$}
\addplot[mark=*, mark size=1.3pt, drainblue] table[x=n, y=b16, col sep=comma] {figures/data/plot_fig4c.csv};
\addlegendentry{$B{=}16$}
\addplot[mark=triangle*, mark size=1.5pt, drainorange] table[x=n, y=b1024, col sep=comma] {figures/data/plot_fig4c.csv};
\addlegendentry{$B{=}1024$}
\draw[black!40, dashed] (axis cs:12,1) -- (axis cs:640,1);
\end{axis}
\end{tikzpicture}
\caption{Support size.}\label{fig:regime}
\end{subfigure}
\caption{\textbf{Budgets, overheads, and scope.} (a), (b) run \texttt{dd-t16} on the indexed path. (a)~POT time over \drain{}, paired per block, as the cap grows. \texttt{dd-t4} fails the common gate at caps 500 and 1000 for both methods and has no eligible ratio there. (b)~Device intervals of the solver call and the common consumer, as shares of the wall time of that run. (c)~OTT-JAX time over \drain{}, paired per block, on synthetic $d=4$ point clouds of equal size in tolerance mode. Below 1, OTT-JAX is faster. $B=1$ was run at $N\in\{16,64,128,512\}$.}
\label{fig:regime-all}
\end{figure}

\Cref{tab:tol} repeats two of the three budget-limited workloads of \cref{tab:main} at four tolerances with the cap raised fivefold, so that fewer than 3\% of instances reach it. \drain{} is \resTolOTTRange$\times$ faster than OTT-JAX and \resTolPOTRange$\times$ faster than POT, whose batch-wide stop runs to the cap at $\tau\le10^{-4}$ on \texttt{dd-t4} and $\tau\le3\cdot10^{-4}$ on \texttt{mot-b16}. Checking every iteration costs at most \resQOne$\times$ the time of checking every 20.

\begin{table}[h]
\centering
\caption{Tolerance sweep. Time per call of \drain{} (ms) and paired ratios per block (median over 5 blocks, range in brackets). $D$: median depth of \drain{}. Capped: instances at the cap. OTT and POT: their time over \drain{}.}
\label{tab:tol}
\small
\setlength{\tabcolsep}{3.5pt}
\resizebox{\linewidth}{!}{%
\begin{tabular}{lrrrrrrll}
\toprule
Workload & $K_{\max}$ & $\tau$ & \drain{} & $D$ & \multicolumn{2}{c}{capped} & OTT & POT \\
\cmidrule(lr){6-7}
& & & & & \drain & OTT & & \\
\midrule
\texttt{dd-t4} & 10000 & $10^{-3}$ & 9.8 & 465 & 0 & 0 & 18.8\,[18.5, 19.2] & 55.4\,[53.9, 56.1] \\
\texttt{dd-t4} & 10000 & $3\cdot10^{-4}$ & 15.8 & 520 & 0 & 0 & 32.6\,[32.5, 32.9] & 93.8\,[90.1, 102.0] \\
\texttt{dd-t4} & 10000 & $10^{-4}$ & 19.9 & 565 & 16 & 16 & 29.5\,[28.9, 32.7] & 81.4\,[79.5, 91.0] \\
\texttt{dd-t4} & 10000 & $3\cdot10^{-5}$ & 20.9 & 620 & 53 & 54 & 28.2\,[27.1, 31.5] & 77.0\,[75.6, 88.4] \\
\texttt{mot-b16} & 2500 & $10^{-3}$ & 5.2 & 30 & 0 & 0 & 16.3\,[16.1, 16.5] & 32.4\,[31.9, 35.4] \\
\texttt{mot-b16} & 2500 & $3\cdot10^{-4}$ & 6.8 & 60 & 0 & 0 & 33.7\,[33.5, 33.8] & 59.6\,[58.9, 60.9] \\
\texttt{mot-b16} & 2500 & $10^{-4}$ & 8.0 & 85 & 3 & 5 & 37.9\,[37.5, 38.0] & 53.9\,[50.4, 55.3] \\
\texttt{mot-b16} & 2500 & $3\cdot10^{-5}$ & 10.0 & 135 & 80 & 69 & 30.3\,[30.3, 30.9] & 41.1\,[40.7, 44.5] \\
 
\bottomrule
\end{tabular}}
\end{table}

\subsection{Depth heterogeneity}
\begin{figure}[t]
\centering
\pgfplotsset{mech/.style={width=\linewidth, height=3.3cm, axis line style={black!60}, tick style={black!60},
  xticklabel style={font=\tiny, rotate=30, anchor=north east, inner sep=1pt}, yticklabel style={font=\tiny},
  label style={font=\scriptsize}, legend style={font=\tiny, draw=none, fill=white, fill opacity=0.8, text opacity=1}}}
\begin{subfigure}[t]{0.325\linewidth}\centering
\begin{tikzpicture}
\begin{axis}[mech, xmin=-0.5, xmax=4.5, ymin=0.7, xtick={0,...,4},
  xticklabels={dd-t4, dd-t8, mot-b16, mot-b32, mot-b64}, xticklabel style={rotate=60, anchor=east}, ylabel={static / active}, ymax=3.2,
  legend style={at={(0.5,1.02)}, anchor=south, legend columns=2}]
\addplot[only marks, mark=-, mark size=0pt, draw=draingray!70, line width=5pt, forget plot,
  error bars/.cd, y dir=both, y explicit, error mark=none, error bar style={draingray!45, line width=5pt}]
  table[x=x, y=mid, y error=half, col sep=comma] {figures/data/plot_fig3a.csv};
\addplot[only marks, mark=*, mark size=1.6pt, drainblue] table[x=x, y=solve, col sep=comma] {figures/data/plot_fig3a.csv};
\addlegendentry{solve interval}
\addplot[only marks, mark=triangle*, mark size=1.8pt, drainorange] table[x=x, y=wall, col sep=comma] {figures/data/plot_fig3a.csv};
\addlegendentry{wall}
\draw[black!40, dashed] (axis cs:-0.5,1) -- (axis cs:4.5,1);
\end{axis}
\end{tikzpicture}
\caption{Release, resident kernel fixed.}\label{fig:a1}
\end{subfigure}\hfill
\begin{subfigure}[t]{0.325\linewidth}\centering
\begin{tikzpicture}
\begin{axis}[mech, ybar, bar width=7pt, xmin=-0.6, xmax=1.6, ymin=0, ymax=3, xtick={0,1},
  xticklabels={dd-t16, mot-win}, ylabel={time / active}, legend style={at={(0.5,1.02)}, anchor=south, legend columns=2}]
\addplot[fill=draingray!60, draw=none] table[x=x, y=masked, col sep=comma] {figures/data/plot_fig3b.csv};
\addlegendentry{masked}
\addplot[fill=drainorange!80, draw=none] table[x=x, y=static, col sep=comma] {figures/data/plot_fig3b.csv};
\addlegendentry{static}
\draw[black!40, dashed] (axis cs:-0.6,1) -- (axis cs:1.6,1);
\end{axis}
\end{tikzpicture}
\caption{Release on the indexed path.}\label{fig:a4}
\end{subfigure}\hfill
\begin{subfigure}[t]{0.325\linewidth}\centering
\begin{tikzpicture}
\begin{axis}[mech, xticklabel style={font=\tiny, rotate=0, anchor=north}, xmin=1, xmax=4, xtick={1,2,3,4}, ymin=0.95, ymax=1.3,
  xlabel={$\Rwork$ of the regrouped batches}, ylabel={static / active}, legend style={at={(0.5,1.02)}, anchor=south, legend columns=2}]
\addplot[only marks, mark=o, mark size=1.6pt, draingray] table[x=rwork, y=ratio, col sep=comma] {figures/data/plot_fig3c_b32.csv};
\addlegendentry{$B=32$}
\addplot[only marks, mark=*, mark size=1.8pt, drainblue] table[x=rwork, y=ratio, col sep=comma] {figures/data/plot_fig3c_b1024.csv};
\addlegendentry{$B=1024$}
\draw[black!40, dashed] (axis cs:1,1) -- (axis cs:4,1);
\end{axis}
\end{tikzpicture}
\caption{Depth heterogeneity and waves.}\label{fig:a3}
\end{subfigure}
\caption{\textbf{Effect of release.} All ratios are paired over \resOuterBlocks{} blocks. Above the dashed line, release is faster. (a)~Resident kernel with release against the same kernel running all instances to $\Dmax$ with the same checks. Gray bars: range predicted by \cref{prop:makespan} with $\sigma=0$. (b)~Indexed path, selected explicitly, with identical iterates: finished instances masked, or all run to $\Dmax$. (c)~Fixed pools of instances regrouped into batches by depth (homogeneous, random, heterogeneous): \texttt{dd-t4}, \texttt{mot-b16}, and \texttt{mot-b32} pools with $B=32$, which fits in one wave, and a \texttt{mot-b16} pool with $B=1024$.}
\label{fig:mechanism}
\end{figure}

Regrouping fixed pools of instances by depth changes $\Rwork$ but no instance (\cref{fig:a3}). With $B=32$, below the 108 SMs of the GPU, every grouping gives \resRegroupSmallB$\times$. With $B=1024$ on \texttt{mot-b16}, the ratio rises with $\Rwork$, from \resRegroupHom$\times$ ($\Rwork=1.73$) to \resRegroupHet$\times$ ($\Rwork=3.70$), although the batch fits in one wave ($S=1296$), where \cref{prop:makespan} predicts no gain. The likely gap is its assumption that programs do not slow each other down: up to 12 programs share an SM, and an exiting program leaves the SM to the others.

\subsection{FP64 audit of returned solutions}
\label{app:audit}
Evaluated in FP64 against each instance's own marginals, every returned solution of every method in \cref{tab:main,tab:structured} passes the gate of $2\cdot10^{-3}$ (worst $\resAuditWorst$). Where instances reach the cap, all methods are alike: on \texttt{dd-t4}, 84.9\% of instances are within $\tau=10^{-4}$ for \drain{}, 84.5\% for OTT-JAX, and 85.1\% for POT, with the same maximum residual.

\subsection{Replay of flow-matching OT batches}
\label{app:replay}
We replayed 27 batches of OT problems captured during CIFAR-10 OT flow-matching training (at capture steps 0, 16, and 32, with 16, 32, or 64 points per instance and 4, 8, or 16 instances per batch) with \resOuterBlocks{} paired blocks per condition. Across conditions, \drain{} is \resReplayOTT$\times$ faster than OTT-JAX and \resReplayPOT$\times$ faster than POT. These are ranges over conditions, not confidence intervals. The replay times the OT calls only.

\section{Training layers}
\label{app:inloop}

\subsection{Captured inputs}
\label{app:capture}
The capture scripts in the released code record every input. This section reports what the captured inputs show.

\paragraph{Label assignment.}
Every FP32 implementation other than \drain{}, the official one included, differs from the FP64 map by up to \resOTADevOthers{} in single plan entries on some calls, and POT changes some assignments on the whole step. \drain{} stays within $\resOTADevDrain$ with the cost shift of \cref{sec:wide}. With every call padded to the largest shape of the run, OTT-JAX takes \resBucketPgMs~ms per call, \resBucketPg$\times$ \drain's time, against 1.05$\times$ with per-call shapes.

\paragraph{mHC.}
We implemented the mHC residual of \citet{xie2025mhc} ($n=4$, 20 iterations) in an 8-block GPT of width 384 (34M parameters), trained it for 2{,}500 steps on FineWeb-edu, and recorded every sublayer's Sinkhorn inputs and upstream gradients at five steps. The logit standard deviation grows from 0.013 at initialization to between 0.3 and 4.2 across sublayers. At step 2{,}500, 44\% of all instances, and 98--100\% in the hardest sublayers, have a two-sided residual above $10^{-4}$ after 20 iterations. \Cref{tab:train}b uses the hardest (std 4.2), a medium (1.6), and an easy (0.3) sublayer, and all 16 sublayers in one call. Peak memory of one forward and backward call above the inputs is \resMemMhcDrain~GB for \drain{} at $B=2^{20}$, against \resMemMhcTorch~GB for the PyTorch reference loop and \resMemMhcCompiled~GB compiled.

\paragraph{Patch matching.}
\drain{} replaces GeoTransformer's fine-matching module behind the same interface, on 32 3DMatch test pairs with the released weights. Estimated poses differ by at most $9.5\cdot10^{-7}$ and losses by $1.2\cdot10^{-7}$ relative. Peak memory above the inputs is \resMemPatchOfficial~MB for the official module and \resMemPatchDrain~MB for \drain{} on each of four pairs.

\subsection{Synthetic mHC inputs}
\Cref{tab:train-s1,tab:train-s10} repeat \cref{tab:train}b with synthetic logits $\sim\mathcal N(0,\sigma^2)$ and up to 1M instances per call.

\begin{table}[h]
\centering
\caption{\cref{tab:train}b with synthetic logits, $\sigma=1$. $^*$Computes a different quantity (\cref{tab:train}).}
\label{tab:train-s1}
\small
\setlength{\tabcolsep}{3pt}
\begin{tabular}{lrrrrrrrrr}
\toprule
& & \multicolumn{4}{c}{forward} & \multicolumn{4}{c}{forward and backward} \\
\cmidrule(lr){3-6}\cmidrule(lr){7-10}
Input & $B$ & ms & compile & triton & TileLang$^*$ & ms & compile & triton & TileLang$^*$ \\
\midrule
$\mathcal N(0,1^2)$, $4\times4$ & 16K & 0.073 & 2.66 & 21.9 & \textbf{0.84} & 0.274 & 2.87 & 7.55 & \textbf{0.31} \\
$\mathcal N(0,1^2)$, $4\times4$ & 128K & 0.265 & 1.53 & 46.7 & 1.07 & 0.763 & 2.37 & 16.7 & \textbf{0.54} \\
$\mathcal N(0,1^2)$, $4\times4$ & 1M & 1.85 & 1.85 & 43.8 & 1.15 & 5.29 & 3.27 & 19.1 & \textbf{0.57} \\
$\mathcal N(0,1^2)$, $8\times8$ & 16K & 0.186 & 1.82 & 43.6 & -- & 0.722 & 1.22 & 16.1 & -- \\
$\mathcal N(0,1^2)$, $8\times8$ & 128K & 1.03 & 1.74 & 51.2 & -- & 4.17 & 1.19 & 22.0 & -- \\
 
\bottomrule
\end{tabular}
\end{table}

\paragraph{mHC with hard inputs.}
\Cref{tab:train-s10} repeats \cref{tab:train}b with $\sigma=10$, the setting of \citet{wang2026birkhoff}. There, at least \resMHCTenUnconverged\% of the instances are above the tolerance after 20 iterations, and TileLang-mHC fails the same-map gate, while the 20-step map keeps its column marginal exact (\drain: $\resMHCTenDrainColumn$).

\begin{table}[h]
\centering
\caption{\cref{tab:train}b with logits $\sim\mathcal N(0,10^2)$. $^\dagger$Fails the same-map gate. $^*$Computes a different quantity (\cref{tab:train}).}
\label{tab:train-s10}
\small
\setlength{\tabcolsep}{3pt}
\begin{tabular}{lrrrrrrrrr}
\toprule
& & \multicolumn{4}{c}{forward} & \multicolumn{4}{c}{forward and backward} \\
\cmidrule(lr){3-6}\cmidrule(lr){7-10}
Input & $B$ & ms & compile & triton & TileLang$^*$ & ms & compile & triton & TileLang$^*$ \\
\midrule
$\mathcal N(0,10^2)$, $4\times4$ & 16K & 0.074 & 2.64 & 21.8 & \textbf{0.83}$^\dagger$ & 0.274 & 2.79 & 7.55 & \textbf{0.30}$^\dagger$ \\
$\mathcal N(0,10^2)$, $4\times4$ & 128K & 0.266 & 1.52 & 46.7 & 1.07$^\dagger$ & 0.761 & 2.38 & 16.7 & \textbf{0.54}$^\dagger$ \\
$\mathcal N(0,10^2)$, $4\times4$ & 1M & 1.85 & 1.85 & 43.8 & 1.15$^\dagger$ & 5.29 & 3.27 & 19.1 & \textbf{0.57}$^\dagger$ \\
$\mathcal N(0,10^2)$, $8\times8$ & 16K & 0.186 & 1.82 & 43.6 & -- & 0.720 & 1.22 & 16.1 & -- \\
$\mathcal N(0,10^2)$, $8\times8$ & 128K & 1.03 & 1.74 & 51.2 & -- & 4.17 & 1.19 & 22.0 & -- \\
 
\bottomrule
\end{tabular}
\end{table}

\paragraph{mHC-proj.}
The $4\times4$ kernel of mHC-proj \citep{wang2026birkhoff} launches on its own stream, which CUDA-graph capture does not record, so it is not compared.

\subsection{Other GPUs}
\label{app:gpus-train}
In the campaign of \cref{app:gpus}, the training layers keep their gains on all three GPUs (\cref{tab:gpus-train}). On DM-Count, the H100 runs the OT loss faster than the A100 but the rest of the step slower, so the loss share and the step gain are smaller there. On the synthetic mHC inputs of \cref{tab:train-s1,tab:train-s10}, \drain{} is \gpuMHCCompileF$\times$ faster than the PyTorch loop under \texttt{torch.compile} in the forward pass and \gpuMHCCompileFB$\times$ with the backward, and \gpuMHCTritonF$\times$ and \gpuMHCTritonFB$\times$ faster than triton-sinkhorn. At $\sigma=10$, TileLang-mHC's forward differs from the 20-step map by up to \gpuMHCTileFwdErr{} and its input gradient by a relative $\gpuMHCTileGradErr$, while \drain{} stays within $\gpuMHCDrainErr$.

\begin{table}[h]
\centering
\caption{Training layers on three GPUs: the official module's time over \drain's: medians of per-block ratios for the OT calls, and over the 20 OTA steps or 32 GeoTransformer pairs for the training steps. Label assignment uses 2 images per call as in \cref{tab:train}a. Patch matching compares forward calls as in \cref{tab:train}c.}
\label{tab:gpus-train}
\small
\begin{tabular}{lrrr}
\toprule
& A100 & H100 & L40S \\
\midrule
Label assignment, 2 images per call & 14.8 & 13.0 & 13.7 \\
\quad OTA training step & 1.17 & 1.15 & 1.09 \\
Patch matching, forward calls & 11.1 & 10.8 & 11.9 \\
\quad GeoTransformer training step & 1.20 & 1.24 & 1.14 \\
Crowd counting, OT loss & 3.52 & 3.25 & 4.43 \\
\quad DM-Count training step & 1.59 & 1.17 & 1.47 \\
 
\bottomrule
\end{tabular}
\end{table}

\subsection{Whole training steps}
\label{app:e2e}
\Cref{tab:e2e} lists the medians over blocks and the ratios paired per block.

\begin{table}[h]
\centering
\caption{Whole training steps (ms, median of 5 blocks). Speedup: official time over the method, paired per block, median [range]. OT: the OT loss called alone on fixed model outputs, forward and backward, ms. In parentheses, its forward call inside the step as a share of the step. $^*$Computes a different quantity from \drain: TileLang-mHC's backward is an implicit gradient at its output (\cref{tab:train}). DM-Count's FP32 solver transports none of the mass of cells whose kernel entries underflow.}
\label{tab:e2e}
\small
\begin{tabular}{llrlr}
\toprule
Model & OT solver & Step & Speedup & OT \\
\midrule
mHC GPT & PyTorch loop (official) & 319 & 1.00 & -- \\
mHC GPT & \texttt{torch.compile} & 243 & 1.32\,[1.30, 1.58] & -- \\
mHC GPT & \drain{} & 184 & 1.73\,[1.71, 2.09] & -- \\
mHC GPT & triton-sinkhorn & 203 & 1.57\,[1.56, 1.89] & -- \\
mHC GPT & TileLang-mHC$^*$ & 177 & 1.80\,[1.79, 2.17] & -- \\
DM-Count & DM-Count loop (official)$^*$ & 455 & 1.00 & 246 (50\%) \\
DM-Count & \drain{} & 302 & 1.51\,[1.47, 1.55] & 73 (24\%) \\
 
\bottomrule
\end{tabular}
\end{table}

\paragraph{mHC GPT.}
The GPT of \cref{app:capture} trains from the same seed on the same tokens under every method. Only the Sinkhorn function changes. The first losses of all methods agree within $\resEtoEMHCStepZero$, and the mean final losses of the five methods lie in \resEtoEMHCFinal.

\paragraph{DM-Count.}
DM-Count \citep{wang2020dmcount} runs 100 Sinkhorn iterations per $512\times512$ crop between the $64\times64$ predicted density and the annotated heads. We replace its per-image FP32 solver by one \drain{} call per batch. In a median \resDMUnderCells\% of the cells, all FP32 kernel entries $e^{-C_{ij}/\eps}$ underflow to zero, so the official solver transports none of their mass (up to \resDMUnderMassMax\% of the density). Against an FP64 run of the same 100-step map, the official density potentials differ by \resDMOfficialFP{} per crop on the remaining cells and those of \drain{} by at most \resDMDrainFP. The two solvers therefore give different density gradients, and we compare their speed only.

\subsection{Training the mHC GPT for 2,500 steps}
\label{app:long}
We trained the GPT of \cref{app:capture} for the full 2{,}500 steps of its capture run, three seeds each with the PyTorch reference loop of mHC, \drain{}, and TileLang-mHC, on one A100 (\cref{fig:long}). The reference loop and \drain{} both compute the gradient of the 20-step map, and all six of their runs end at a held-out loss of \longExactRange. TileLang-mHC follows them for about 200 steps and then stalls, ending at \longTileRange. Its last logged gradient norm before clipping is \longTileGrad, against \longExactGrad{} for the other two.

\begin{figure}[h]
\centering
\begin{tikzpicture}
\begin{axis}[width=0.75\linewidth, height=4.2cm, axis line style={black!60}, tick style={black!60},
  xticklabel style={font=\tiny}, yticklabel style={font=\tiny}, label style={font=\scriptsize},
  xlabel={training step}, ylabel={held-out loss}, xmin=0, xmax=2500, ymin=4.3, ymax=8,
  legend style={font=\tiny, draw=none, at={(0.98,0.98)}, anchor=north east}]
\pgfplotsinvokeforeach{0,1,2}{
  \addplot[draingreen, thick, forget plot] table[x=step, y=pytorch#1, col sep=comma] {figures/data/plot_mhc_long.csv};
  \addplot[drainblue, thick, dashed, forget plot] table[x=step, y=drain#1, col sep=comma] {figures/data/plot_mhc_long.csv};
  \addplot[drainorange, thick, forget plot] table[x=step, y=tilelang#1, col sep=comma] {figures/data/plot_mhc_long.csv};
}
\addlegendimage{draingreen, thick}\addlegendentry{PyTorch reference loop}
\addlegendimage{drainblue, thick, dashed}\addlegendentry{\drain}
\addlegendimage{drainorange, thick}\addlegendentry{TileLang-mHC$^*$}
\end{axis}
\end{tikzpicture}
\caption{Held-out loss of the mHC GPT over 2{,}500 steps, three seeds per method, evaluated every 100 steps. $^*$Computes a different quantity (\cref{tab:train}).}
\label{fig:long}
\end{figure}

\end{document}